\documentclass[a4paper,12pt]{article}
\usepackage{amssymb}
\usepackage{amsmath}
\usepackage{amsthm}
\usepackage[utf8]{inputenc}
\usepackage[T1]{fontenc}
\usepackage{graphicx}
\usepackage{tikz}
\usepackage{pgf,tikz}
\usepackage{enumitem}
\usepackage{color}
\usetikzlibrary{arrows}

\newtheorem{theo}{Theorem}
\newtheorem{cor}[theo]{Corollary}
\newtheorem{lemm}[theo]{Lemma}

\newtheorem{obs}[theo]{Observation}
\newtheorem{conj}[theo]{Conjecture}

\title{Large induced forests in planar graphs with girth 4 or 5}

\usepackage{hyperref}
\usepackage[affil-it]{authblk}

\author[a]{François Dross}
\author[b]{Mickael Montassier}
\author[c]{Alexandre Pinlou}

\affil[a]{{\small ENS de Lyon, LIRMM}}
\affil[b]{{\small Université Montpellier 2, LIRMM}}
\affil[c]{{\small Université Montpellier 3, LIRMM\medskip}}
\affil[ ]{{\small 161 rue Ada, 34095 Montpellier Cedex 5, France}}
\affil[
]{\href{mailto:francois.dross@ens-lyon.fr,mickael.montassier@lirmm.fr,alexandre.pinlou@lirmm.fr}{\small{francois.dross@ens-lyon.fr,\{mickael.montassier,alexandre.pinlou\}@lirmm.fr}}}

\begin{document}
\definecolor{yqyqyq}{rgb}{0.5019607843137255,0.5019607843137255,0.5019607843137255}
\definecolor{uuuuuu}{rgb}{0.26666666666666666,0.26666666666666666,0.26666666666666666}

\maketitle

\begin{abstract}
We give here some new lower bounds on the order of a largest induced
forest in planar graphs with girth $4$ and $5$. In particular we prove
that a triangle-free planar graph of order $n$ admits an induced
forest of order at least $\frac{6n + 7}{11}$, improving the lower
bound of Salavatipour [M.~R.~Salavatipour, Large induced forests in
  triangle-free planar graphs, {\em Graphs and Combinatorics},
  22:113--126, 2006]. We also prove that a planar graph of order $n$
and girth at least $5$ admits an induced forest of order at least
$\frac{44n+50}{69}$.
\end{abstract}
\section{Introduction}

Let $G$ be a graph. A \emph{decycling set} or \emph{feedback vertex set} $S$ of $G$ is a subset of the vertices of $G$ such that removing the vertices of $S$ from $G$ yields an acyclic graph. Thus $S$ is a decycling set of $G$ if and only if the graph induced by $V(G) \backslash S$ in $G$ is an induced forest of $G$. The {\sc feedback vertex set decision problem} (which consists of, given a graph $G$ and an integer $k$, deciding whether there is a decycling set of $G$ of size $k$) is known to be NP-complete, even restricted to the case of planar graphs, bipartite graphs or perfect graphs \cite{Karp}. It is thus legitimate to seek bounds for the size of a decycling set or an induced forest. The smallest size of a decycling set of $G$ is called the \emph{decycling number} of $G$, and the highest order of an induced forest of $G$ is called the \emph{forest number} of $G$, denoted respectively by $\phi(G)$ and $a(G)$. Note that the sum of the decycling number and the forest number of $G$ is equal to the order of $G$ (i.e. $|V(G)| = a(G) + \phi(G)$).

Mainly, the community focuses on the following challenging conjecture due to Albertson and Berman \cite{AlbertsonBerman}:

\begin{conj}[\emph{Albertson and Berman \cite{AlbertsonBerman}}] \label{alb}
Every planar graph of order $n$ admits an induced forest of order at least $\frac{n}{2}$.
\end{conj}
Conjecture~\ref{alb}, if true, would be tight (for $n \ge 3$ multiple of $4$) because of the disjoint union of the complete graph on four vertices (Akiyama and Watanabe \cite{Akiyama} gave examples showing that the conjecture differs from the optimal by at most one half for all $n$), and would imply that every planar graph has an independent set on at least a quarter of its vertices, the only known proof of which relies on the Four-Color Theorem. 

The best known lower bound to date for the forest number of a planar graph is due to Borodin and is a consequence of the acyclic $5$-colorability of planar graphs \cite{Borodin}. We recall that an acyclic coloring is a proper vertex coloring such that the graph induced by the vertices of any two color classes is a forest. From this result we obtain the following theorem:

\begin{theo}[\emph{Borodin \cite{Borodin}}]
Every planar graph of order $n$ admits an induced forest of order at least $\frac{2n}{5}$.
\end{theo}

Hosono \cite{Hosono} showed the following theorem as a consequence of the acyclic 3-colorability of outerplanar graphs and showed that the bound is tight.

\begin{theo}[\emph{Hosono \cite{Hosono}}] \label{hos}
Every outerplanar graph of order $n$ admits an induced forest of order at least $\frac{2n}{3}$.
\end{theo}
The tightness of the bound is shown by the example in Figure~\ref{bababibelba}.

\begin{figure}[h]
\begin{center}
\begin{tikzpicture}[line cap=round,line join=round,>=triangle 45,x=1.0cm,y=1.0cm]
\clip(-0.25,-0.25) rectangle (8.25,1.25);
\draw (0.0,-0.0)-- (0.0,1.0);
\draw (1.0,0.0)-- (1.0,1.0);
\draw (2.0,0.0)-- (2.0,1.0);
\draw (3.0,0.0)-- (3.0,1.0);
\draw (7.0,0.0)-- (7.0,1.0);
\draw (8.0,0.0)-- (8.0,1.0);
\draw (0.0,1.0)-- (1.0,1.0);
\draw (1.0,1.0)-- (2.0,1.0);
\draw (2.0,1.0)-- (3.0,1.0);
\draw (0.0,-0.0)-- (1.0,-0.0);
\draw (1.0,-0.0)-- (2.0,0.0);
\draw (2.0,0.0)-- (3.0,0.0);
\draw (0.0,1.0)-- (1.0,-0.0);
\draw (1.0,1.0)-- (2.0,0.0);
\draw (2.0,1.0)-- (3.0,0.0);
\draw [dash pattern=on 1pt off 1pt] (3.0,1.0)-- (7.0,1.0);
\draw [dash pattern=on 1pt off 1pt] (3.0,0.0)-- (7.0,0.0);
\draw (7.0,1.0)-- (8.0,1.0);
\draw (7.0,0.0)-- (8.0,0.0);
\draw (7.0,1.0)-- (8.0,0.0);
\begin{scriptsize}
\draw [fill=black] (0.0,-0.0) circle (1.5pt);
\draw [fill=black] (0.0,1.0) circle (1.5pt);
\draw [fill=black] (1.0,-0.0) circle (1.5pt);
\draw [fill=black] (1.0,0.0) circle (1.5pt);
\draw [fill=black] (1.0,1.0) circle (1.5pt);
\draw [fill=black] (2.0,0.0) circle (1.5pt);
\draw [fill=black] (2.0,1.0) circle (1.5pt);
\draw [fill=black] (3.0,0.0) circle (1.5pt);
\draw [fill=black] (3.0,1.0) circle (1.5pt);
\draw [fill=black] (7.0,0.0) circle (1.5pt);
\draw [fill=black] (7.0,1.0) circle (1.5pt);
\draw [fill=black] (8.0,0.0) circle (1.5pt);
\draw [fill=black] (8.0,1.0) circle (1.5pt);
\end{scriptsize}
\end{tikzpicture}
\end{center}
\caption{Example to prove the tightness of Theorem~\ref{hos}. \label{bababibelba}}
\end{figure}
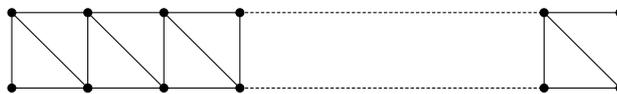

Other results were deduced from results on acyclic coloring, for other classes of graphs. Fertin et al. \cite{Fertin} gave such results for several classes of graphs, stated in Table~\ref{matable}.

\begin{table}[h]
\begin{centering}
\begin{tabular}{|c|c|c|}
        \hline
        Family $\cal{F}$ & \multicolumn{2}{c|}{Forest number:} \\
 
         &Lower bound & Upper bound \\ 
        \hline && \\
        Planar & $\frac{2n}{5}$ & $\lceil\frac{n}{2}\rceil$ \\ && \\
        Planar with girth $5$, $6$ & $\frac{n}{2}$ & $\frac{7n}{10} + 2$ \\ && \\
        Planar with girth $\ge 7$ & $\frac{2n}{3}$ & $\frac{5n}{6} + 1$ \\ && \\
        \hline
\end{tabular}
   \caption{\label{matable} Bounds on the forest number for some families $\cal{F}$ of graphs~\cite{Fertin}.}
\end{centering}
\end{table}
\bigskip 

Akiyama and Watanabe \cite{Akiyama}, and Albertson and Rhaas \cite{Albertson} independently raised the following conjecture:
\begin{conj}[\emph{Akiyama and Watanabe \cite{Akiyama}, and Albertson and Rhaas \cite{Albertson}}] \label{aki}
Every bipartite planar graph of order $n$ admits an induced forest of order at least $\frac{5n}{8}$.
\end{conj}
This conjecture, if true, would be tight for $n$ multiple of $8$: for example if $G$ is the disjoint union of $k$ cubes, then we have $a(G) = 5k$ and $G$ has order $8k$ (see Figure~\ref{bababa}). Motivated by Conjecture~\ref{aki}, Alon~\cite{Alon2003} proved the following theorem using probabilistic methods:

\begin{figure}[h]
\begin{center}
\definecolor{cqcqcq}{rgb}{0.65,0.65,0.65}
\begin{tikzpicture}[line cap=round,line join=round,>=triangle 45,x=0.5cm,y=0.5cm]
\draw (0.6879722606657073,6.03171952230884)-- (5.999242097886516,6.031719522308842);
\draw [color=cqcqcq] (5.999242097886516,6.031719522308842)-- (5.999242097886509,0.7204496850880373);
\draw [color=cqcqcq] (5.999242097886509,0.7204496850880373)-- (0.6879722606657072,0.7204496850880311);
\draw (0.6879722606657072,0.7204496850880311)-- (0.6879722606657073,6.03171952230884);
\draw (0.6879722606657073,6.03171952230884)-- (2.243607179276111,4.476084603698435);
\draw [color=cqcqcq] (2.243607179276111,4.476084603698435)-- (4.4436071792761105,4.476084603698435);
\draw [color=cqcqcq] (4.4436071792761105,4.476084603698435)-- (4.443607179276109,2.276084603698438);
\draw [color=cqcqcq] (4.443607179276109,2.276084603698438)-- (2.243607179276111,2.2760846036984357);
\draw [color=cqcqcq] (2.243607179276111,2.2760846036984357)-- (2.243607179276111,4.476084603698435);
\draw [color=cqcqcq] (2.243607179276111,2.2760846036984357)-- (0.6879722606657072,0.7204496850880311);
\draw [color=cqcqcq] (4.443607179276109,2.276084603698438)-- (5.999242097886509,0.7204496850880373);
\draw [color=cqcqcq] (4.4436071792761105,4.476084603698435)-- (5.999242097886516,6.031719522308842);
\begin{scriptsize}
\draw [fill=black] (2.243607179276111,4.476084603698435) circle (1.5pt);
\draw [color=cqcqcq][fill=cqcqcq] (2.243607179276111,2.2760846036984357) circle (1.5pt);
\draw [color=cqcqcq][fill=cqcqcq] (4.4436071792761105,4.476084603698435) circle (1.5pt);
\draw [fill=black] (4.443607179276109,2.276084603698438) circle (1.5pt);
\draw [fill=black] (0.6879722606657073,6.03171952230884) circle (1.5pt);
\draw [fill=black] (5.999242097886516,6.031719522308842) circle (1.5pt);
\draw [color=cqcqcq][fill=cqcqcq] (5.999242097886509,0.7204496850880373) circle (1.5pt);
\draw [fill=black] (0.6879722606657072,0.7204496850880311) circle (1.5pt);
\end{scriptsize}
\end{tikzpicture}
\end{center}
\caption{The cube admits an induced forest on five of its vertices, but no induced forest on six or more of its vertices. \label{bababa}}
\end{figure}

\begin{theo}[\emph{Alon~\cite{Alon2003}}]
There exist some $b > 0$ and $b' > 0$ such that:

\begin{itemize}
\item For every bipartite graph $G$ with $n$ vertices and average degree at most $d$ ($\ge 1$), $a(G) \ge (\frac{1}{2} + e^{-bd^2})n$.

\item For every $d \ge 1$ and all sufficiently large $n$ there exists a bipartite graph with $n$ vertices and average degree at most $d$ such that  $a(G) \le (\frac{1}{2} + e^{-b'\sqrt{d}})n$.
\end{itemize}
\end{theo}
The lower bound was later improved by Colon et al. \cite{conlonessays} to $a(G) \ge (1/2 + e^{-b''d})n$ for a constant $b''$. 

Conjecture~\ref{aki} also led to some research for lower bounds of the forest number of triangle-free planar graphs (as a superclass of bipartite planar graphs). Alon et al.~\cite{Alon} proved the following theorems and corollary: 

\begin{theo}[\emph{Alon et al.~\cite{Alon}}]\label{bAlon}
Every triangle-free graph of order $n$ and size $m$ admits an induced forest of order at least $n - \frac{m}{4}$.
\end{theo}

\begin{cor}[\emph{Alon et al.~\cite{Alon}}]
Every triangle-free cubic graph of order $n$ admits an induced forest of order at least $\frac{5n}{8}$.
\end{cor}

\begin{theo}[\emph{Alon et al.~\cite{Alon}}]
Every connected graph with maximum degree $\Delta$, order $n$, and size $m$ admits an induced forest of order at least $\alpha(G) + \frac{n-\alpha(G)}{(\Delta - 1)^2}$.
\end{theo}

Theorem~\ref{bAlon} is tight because of the union of cycles of length $4$.

In a planar graph with girth at least $g$, order $n$ and size $m$ with at least a cycle, the number of faces is at most $2m/g$ (since all the faces' boundaries have length at least $g$). Then, by Euler's formula, $2m/g \ge m - n + 2$, and thus $m \le (g/(g-2)) (n - 2)$. In particular, triangle-free planar graphs of order $n \ge 3$ have size at most $2n-4$. 

As a consequence of Theorem~\ref{bAlon}, for $G$ a triangle-free planar graph of order $n$, $a(G) \ge n/2$. This lower bound was improved for $n \ge 1$ by Salavatipour \cite{Salavatipour}.

\begin{theo}[\emph{Salavatipour~\cite{Salavatipour}}]\label{salavat}
Every triangle-free planar graph of order $n$ and size $m$ admits an induced forest of order at least $\frac{29n-6m}{32}$ and thus at least $\frac{17n + 24}{32}$.
\end{theo}

In $2010$, Kowalik et al. \cite{Kowalik} proposed that for triangle-free planar graphs of order $n$ and size $m$, $a(G) \ge \frac{119n -24m -24}{128} \ge \frac{71n +72}{128}$. However, it seems that the proof has a flaw. We give here an infinite family of counter-examples for $a(G) \ge \frac{119n -24m -24}{128}$ (see Section~\ref{proofmain}). We propose an improvement of Theorem~\ref{salavat}, which thus leads to the best known bound to our knowledge (see Section~\ref{proofmain}):

\begin{theo} \label{main}
Every triangle-free planar graph of order $n$ and size $m$ admits an induced forest of order at least $\max\{\frac{38n-7m}{44},n-\frac{m}4\}$.
\end{theo}

Hence by Euler's formula the following corollary holds:

\begin{cor} \label{comain}
Every triangle-free planar graph of order $n \ge 1$ admits an induced forest of order at least $\frac{6n + 7}{11}$.
\end{cor}

Kowalik et al.~\cite{Kowalik} made the following conjecture on planar graph of girth at least $5$:
\begin{conj}[\emph{Kowalik et al.~\cite{Kowalik}}] \label{conj:kowalik}
Every planar graph with girth at least $5$ and order $n$ admits an induced forest of order at least $7n/10$.
\end{conj}
This conjecture, if true, would be tight for $n$ multiple of $20$, as shown by the example of the union of dodecahedron, given by Kowalik et al.~\cite{Kowalik} (see Figure~\ref{bibibi}).
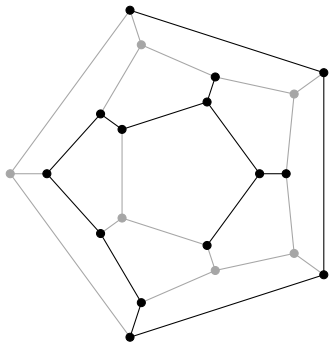
\begin{figure}[h]
\begin{center}
\definecolor{cqcqcq}{rgb}{0.65,0.65,0.65}
\begin{tikzpicture}[line cap=round,line join=round,>=triangle 45,x=1.0cm,y=1.0cm]
\clip(-2.6,-2.5) rectangle (2.3,2.5);
\draw (1.0,-0.0)-- (0.3090169943749475,0.9510565162951535);
\draw (0.3090169943749475,0.9510565162951535)-- (-0.8090169943749473,0.5877852522924731);
\draw [color=cqcqcq] (-0.8090169943749473,0.5877852522924731)-- (-0.8090169943749475,-0.5877852522924731);
\draw (0.30901699437494745,-0.9510565162951536)-- (1.0,-0.0);
\draw [color=cqcqcq] (0.30901699437494745,-0.9510565162951536)-- (-0.8090169943749475,-0.5877852522924731);
\draw [color=cqcqcq] (1.4543870332530722,1.056674031837904)-- (1.8443253811383944,1.3399808248767398);
\draw (1.8443253811383944,1.3399808248767398)-- (-0.7044696092807624,2.1681345189236856);
\draw [color=cqcqcq] (-0.7044696092807624,2.1681345189236856)-- (-0.55552641390555,1.709734498543117);
\draw [color=cqcqcq] (-0.7044696092807624,2.1681345189236856)-- (-2.279711543715264,0.0);
\draw [color=cqcqcq] (-2.279711543715264,0.0)-- (-0.7044696092807627,-2.1681345189236856);
\draw (-0.7044696092807627,-2.1681345189236856)-- (1.8443253811383948,-1.3399808248767402);
\draw (1.8443253811383948,-1.3399808248767402)-- (1.8443253811383944,1.3399808248767398);
\draw [color=cqcqcq] (1.4543870332530722,1.056674031837904)-- (0.41694269437650827,1.2832176664280721);
\draw [color=cqcqcq] (0.41694269437650827,1.2832176664280721)-- (-0.55552641390555,1.709734498543117);
\draw [color=cqcqcq] (-0.55552641390555,1.709734498543117)-- (-1.091570145238658,0.7930721328168736);
\draw (-1.091570145238658,0.7930721328168736)-- (-1.7977212386950445,0.0);
\draw (-1.7977212386950445,0.0)-- (-1.091570145238658,-0.7930721328168735);
\draw (-1.091570145238658,-0.7930721328168735)-- (-0.5555264139055502,-1.709734498543117);
\draw [color=cqcqcq] (-0.5555264139055502,-1.709734498543117)-- (0.4169426943765081,-1.2832176664280721);
\draw [color=cqcqcq] (0.4169426943765081,-1.2832176664280721)-- (1.4543870332530724,-1.0566740318379042);
\draw [color=cqcqcq] (1.4543870332530724,-1.0566740318379042)-- (1.3492549017242996,-0.0);
\draw [color=cqcqcq] (1.3492549017242996,-0.0)-- (1.4543870332530722,1.056674031837904);
\draw (1.3492549017242996,-0.0)-- (1.0,-0.0);
\draw (0.3090169943749475,0.9510565162951535)-- (0.41694269437650827,1.2832176664280721);
\draw (-1.091570145238658,0.7930721328168736)-- (-0.8090169943749473,0.5877852522924731);
\draw [color=cqcqcq] (-1.7977212386950445,0.0)-- (-2.279711543715264,0.0);
\draw [color=cqcqcq] (-1.091570145238658,-0.7930721328168735)-- (-0.8090169943749475,-0.5877852522924731);
\draw (-0.5555264139055502,-1.709734498543117)-- (-0.7044696092807627,-2.1681345189236856);
\draw [color=cqcqcq] (0.4169426943765081,-1.2832176664280721)-- (0.30901699437494745,-0.9510565162951536);
\draw [color=cqcqcq] (1.4543870332530724,-1.0566740318379042)-- (1.8443253811383948,-1.3399808248767402);
\begin{scriptsize}
\draw [fill=black] (0.3090169943749475,0.9510565162951535) circle (1.5pt);
\draw [fill=black] (1.0,-0.0) circle (1.5pt);
\draw [fill=black] (-0.8090169943749473,0.5877852522924731) circle (1.5pt);
\draw [color=cqcqcq][fill=cqcqcq] (-0.8090169943749475,-0.5877852522924731) circle (1.5pt);
\draw [fill=black] (0.30901699437494745,-0.9510565162951536) circle (1.5pt);
\draw [fill=black] (-1.091570145238658,0.7930721328168736) circle (1.5pt);
\draw [fill=black] (0.41694269437650827,1.2832176664280721) circle (1.5pt);
\draw [fill=black] (1.3492549017242996,-0.0) circle (1.5pt);
\draw [color=cqcqcq][fill=cqcqcq] (0.4169426943765081,-1.2832176664280721) circle (1.5pt);
\draw [fill=black] (-1.091570145238658,-0.7930721328168735) circle (1.5pt);
\draw [color=cqcqcq][fill=cqcqcq] (1.4543870332530722,1.056674031837904) circle (1.5pt);
\draw [color=cqcqcq][fill=cqcqcq] (-0.55552641390555,1.709734498543117) circle (1.5pt);
\draw [color=cqcqcq][fill=cqcqcq] (1.4543870332530724,-1.0566740318379042) circle (1.5pt);
\draw [fill=black] (-0.5555264139055502,-1.709734498543117) circle (1.5pt);
\draw [fill=black] (-1.7977212386950445,0.0) circle (1.5pt);
\draw [fill=black] (1.8443253811383944,1.3399808248767398) circle (1.5pt);
\draw [fill=black] (-0.7044696092807624,2.1681345189236856) circle (1.5pt);
\draw [color=cqcqcq][fill=cqcqcq] (-2.279711543715264,0.0) circle (1.5pt);
\draw [fill=black] (-0.7044696092807627,-2.1681345189236856) circle (1.5pt);
\draw [fill=black] (1.8443253811383948,-1.3399808248767402) circle (1.5pt);
\end{scriptsize}
\end{tikzpicture}
\end{center}
\caption{The dodecahedron admits an induced forest on fourteen of its vertices, but no induced forest on fifteen or more of its vertices. \label{bibibi}}
\end{figure}
We prove the following theorem which is a first step toward Conjecture~\ref{conj:kowalik} (see Section~\ref{bproofmain}):

\begin{theo} \label{bmain}
Every planar graph with girth at least $5$, order $n$ and size $m$ admits an induced forest of order at least $n - \frac{5m}{23}$.
\end{theo}

Hence by Euler's formula the following corollary holds:

\begin{cor} \label{bcomain}
Every planar graph with girth at least $5$ and order $n\ge 1$ admits an induced forest of order at least $\frac{44n+50}{69}$. 
\end{cor}

From Theorem~\ref{bmain} we can deduce, with Euler's formula (which implies that $m \le (g/(g-2)) (n - 2)$), the following corollary:

\begin{cor} \label{bcomainbis}
Every planar graph with girth at least $g \ge 5$ and order $n\ge 1$ admits an induced forest of order at least $n - \frac{(5n-10)g}{23(g-2)}$.
\end{cor}

\begin{table}[h!]
\begin{centering}
\begin{tabular}{|c|c|c|}
        \hline
        Girth higher than & Lower bound for $a(G)$ & $a(G)$ for a graph of this class \\
        \hline && \\

        $4$ & $\frac{6n + 7}{11}$ & $\frac{5n}{8}$ \\ && \\

        $5$ & $\frac{44n+50}{69}$ & $\frac{7n}{10}$ \\ && \\

        $6$ & $\frac{31n + 30}{46}$ & $\frac{23n}{30}$ \\ && \\

        $7$ & $\frac{16n + 14}{23}$ & $\frac{17n}{21}$ \\ && \\
        \hline
\end{tabular}
   \caption{\label{matablebis} Our lower bounds on $a(G)$ for $G$ planar graph of high enough girth, compared to the best possible lower bounds for $a(G)$ on the corresponding classes of graphs.}
\end{centering}
\end{table}

Finally, we summarize lower and upper bounds in Table~\ref{matablebis}. The upper bounds for girth $6$ and $7$ are obtained by the graphs in Figures~\ref{bilibilibi} and~\ref{bolobolobo}. There is no bigger induced forest for any of them since all vertices have degree at most $3$, and thus at least one vertex per two faces have to be removed.

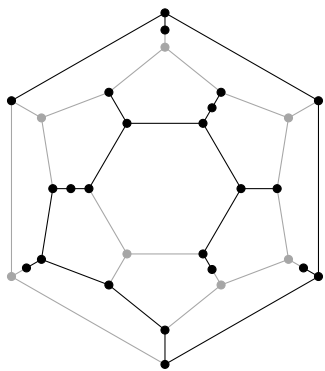
\begin{figure}[h!]
\begin{center}
\definecolor{cqcqcq}{rgb}{0.65,0.65,0.65}
\begin{tikzpicture}[line cap=round,line join=round,>=triangle 45,x=1.0cm,y=1.0cm]
\clip(-2.5,-2.5) rectangle (2.5,2.5);
\draw (-0.5,0.8660254037844386)-- (-1.0,0);
\draw [color=cqcqcq] (-1.0,0)-- (-0.4999999999999999,-0.8660254037844386);
\draw [color=cqcqcq] (-0.4999999999999999,-0.8660254037844386)-- (0.5,-0.8660254037844387);
\draw (0.5,-0.8660254037844387)-- (1.0,0.0);
\draw (0.5,0.8660254037844386)-- (1.0,0.0);
\draw (0.5,0.8660254037844386)-- (-0.5,0.8660254037844386);
\draw (-0.5,0.8660254037844386)-- (-0.7376958531495094,1.2777266981878195);
\draw (-1.0,0)-- (-1.4753917062990187,0);
\draw [color=cqcqcq] (-0.4999999999999999,-0.8660254037844386)-- (-0.7376958531495095,-1.2777266981878197);
\draw (1.0,0.0)-- (1.4753917062990187,0);
\draw (0.5,0.8660254037844386)-- (0.7376958531495094,1.2777266981878195);
\draw [color=cqcqcq] (-0.7376958531495094,1.2777266981878195)-- (0.0,1.8757463842996907);
\draw [color=cqcqcq] (0.0,1.8757463842996907)-- (0.7376958531495094,1.2777266981878195);
\draw (0.0,2.33100325520568)-- (-2.01870803531234,1.16550162760284);
\draw [color=cqcqcq] (-2.01870803531234,1.16550162760284)-- (-1.62444401986034,0.9378731921498452);
\draw [color=cqcqcq] (-1.62444401986034,0.9378731921498452)-- (-0.7376958531495094,1.2777266981878195);
\draw [color=cqcqcq] (-1.62444401986034,0.9378731921498452)-- (-1.4753917062990187,0);
\draw [color=cqcqcq] (-2.01870803531234,1.16550162760284)-- (-2.01870803531234,-1.16550162760284);
\draw (-1.4753917062990187,0)-- (-1.62444401986034,-0.9378731921498451);
\draw [color=cqcqcq] (-2.01870803531234,-1.16550162760284)-- (0.0,-2.3310032552056805);
\draw (0.0,-1.8757463842996902)-- (-0.7376958531495095,-1.2777266981878197);
\draw (-0.7376958531495095,-1.2777266981878197)-- (-1.62444401986034,-0.9378731921498451);
\draw (0.0,-1.8757463842996902)-- (0.0,-2.3310032552056805);
\draw [color=cqcqcq] (0.0,-1.8757463842996902)-- (0.7376958531495095,-1.2777266981878197);
\draw [color=cqcqcq] (0.7376958531495095,-1.2777266981878197)-- (1.62444401986034,-0.9378731921498452);
\draw (0.0,-2.3310032552056805)-- (2.01870803531234,-1.16550162760284);
\draw (2.01870803531234,-1.16550162760284)-- (2.01870803531234,1.16550162760284);
\draw (2.01870803531234,1.16550162760284)-- (0.0,2.33100325520568);
\draw [color=cqcqcq] (2.01870803531234,1.16550162760284)-- (1.62444401986034,0.9378731921498451);
\draw [color=cqcqcq] (1.62444401986034,0.9378731921498451)-- (0.7376958531495094,1.2777266981878195);
\draw [color=cqcqcq] (1.62444401986034,0.9378731921498451)-- (1.4753917062990187,0);
\draw [color=cqcqcq] (1.4753917062990187,0)-- (1.62444401986034,-0.9378731921498452);
\draw [color=cqcqcq] (-2.01870803531234,-1.16550162760284)-- (-1.82157602758634,-1.0516874098763425);
\draw (-1.82157602758634,-1.0516874098763425)-- (-1.62444401986034,-0.9378731921498451);
\draw [color=cqcqcq] (0.0,2.1033748197526854)-- (0.0,1.8757463842996907);
\draw (0.0,2.33100325520568)-- (0.0,2.1033748197526854);
\draw [color=cqcqcq] (1.62444401986034,-0.9378731921498452)-- (1.82157602758634,-1.0516874098763427);
\draw (1.82157602758634,-1.0516874098763427)-- (2.01870803531234,-1.16550162760284);
\draw (0.5,-0.8660254037844387)-- (0.6188479265747547,-1.0718760509861291);
\draw [color=cqcqcq] (0.6188479265747547,-1.0718760509861291)-- (0.7376958531495095,-1.2777266981878197);
\begin{scriptsize}
\draw [fill=black] (1.0,0.0) circle (1.5pt);
\draw [fill=black] (0.5,0.8660254037844386) circle (1.5pt);
\draw [fill=black] (-0.5,0.8660254037844386) circle (1.5pt);
\draw [fill=black] (-1.0,0) circle (1.5pt);
\draw [color=cqcqcq][fill=cqcqcq] (-0.4999999999999999,-0.8660254037844386) circle (1.5pt);
\draw [fill=black] (0.5,-0.8660254037844387) circle (1.5pt);
\draw [fill=black] (1.4753917062990187,0) circle (1.5pt);
\draw [color=cqcqcq][fill=cqcqcq] (1.62444401986034,0.9378731921498451) circle (1.5pt);
\draw [fill=black] (2.01870803531234,1.16550162760284) circle (1.5pt);
\draw [color=cqcqcq][fill=cqcqcq] (1.62444401986034,-0.9378731921498452) circle (1.5pt);
\draw [fill=black] (2.01870803531234,-1.16550162760284) circle (1.5pt);
\draw [fill=black] (0.0,-1.8757463842996902) circle (1.5pt);
\draw [fill=black] (0.0,-2.3310032552056805) circle (1.5pt);
\draw [color=cqcqcq][fill=cqcqcq] (0.7376958531495095,-1.2777266981878197) circle (1.5pt);
\draw [fill=black] (-0.7376958531495095,-1.2777266981878197) circle (1.5pt);
\draw [color=cqcqcq][fill=cqcqcq] (-2.01870803531234,-1.16550162760284) circle (1.5pt);
\draw [fill=black] (-1.62444401986034,-0.9378731921498451) circle (1.5pt);
\draw [color=cqcqcq][fill=cqcqcq] (-1.62444401986034,0.9378731921498452) circle (1.5pt);
\draw [fill=black] (-2.01870803531234,1.16550162760284) circle (1.5pt);
\draw [fill=black] (-1.4753917062990187,0) circle (1.5pt);
\draw [fill=black] (-0.7376958531495094,1.2777266981878195) circle (1.5pt);
\draw [fill=black] (0.0,2.33100325520568) circle (1.5pt);
\draw [color=cqcqcq][fill=cqcqcq] (0.0,1.8757463842996907) circle (1.5pt);
\draw [fill=black] (0.7376958531495094,1.2777266981878195) circle (1.5pt);
\draw [fill=black] (0.0,2.1033748197526854) circle (1.5pt);
\draw [fill=black] (-1.82157602758634,-1.0516874098763425) circle (1.5pt);
\draw [fill=black] (1.82157602758634,-1.0516874098763427) circle (1.5pt);
\draw [fill=black] (0.6188479265747547,-1.0718760509861291) circle (1.5pt);
\draw [fill=black] (-1.2376958531495093,0) circle (1.5pt);
\draw [fill=black] (0.6188479265747546,1.0718760509861291) circle (1.5pt);
\end{scriptsize}
\end{tikzpicture}
\end{center}
\caption{A planar graph of girth $6$ on $30$ vertices that admits an induced forest on $23$ of its vertices, but no induced forest on $24$ or more of its vertices. \label{bilibilibi}}
\end{figure}

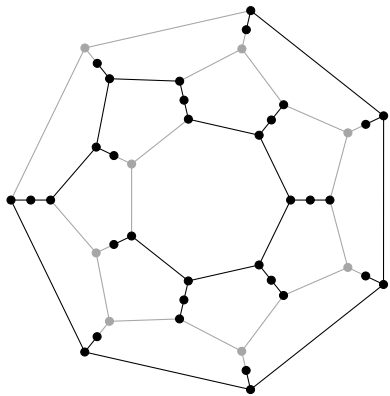
\begin{figure}[h!]
\begin{center}
\definecolor{cqcqcq}{rgb}{0.65,0.65,0.65}
\begin{tikzpicture}[line cap=round,line join=round,>=triangle 45,x=0.55cm,y=0.55cm]
\clip(-5.5,-5.5) rectangle (5.5,5.5);
\draw [color=cqcqcq] (-2.9187046762420463,3.6666867736051882)-- (-2.6238107350496556,3.2962197912523714);
\draw (-2.6238107350496556,3.2962197912523714)-- (-2.328916793857265,2.9257528088995546);
\draw (-0.6533539251791053,2.8664836731530596)-- (-0.548906188840949,2.4082362832272306);
\draw (-0.548906188840949,2.4082362832272306)-- (-0.44445845250279276,1.9499888933014011);
\draw (1.0462651612432865,4.568233585886012)-- (0.9405548235572622,4.106678013856067);
\draw [color=cqcqcq] (0.9405548235572622,4.106678013856067)-- (0.8348444858712378,3.645122441826121);
\draw (1.8334038643972352,2.2983103075984506)-- (1.5403086887963167,1.9308933536626438);
\draw (1.5403086887963167,1.9308933536626438)-- (1.247213513195398,1.5634763997268373);
\draw (4.223620723730235,2.0308758861449157)-- (3.796883420891386,1.8256845219713524);
\draw [color=cqcqcq] (3.796883420891386,1.8256845219713524)-- (3.370146118052537,1.6204931577977892);
\draw (2.94,0.0)-- (2.4699999999999998,0.0);
\draw (2.4699999999999998,0.0)-- (2.0,0.0);
\draw [color=cqcqcq] (3.3684472940068986,-1.6240214773800357)-- (3.794969487006044,-1.829659607221715);
\draw (3.794969487006044,-1.829659607221715)-- (4.221491680005189,-2.0352977370633947);
\draw (1.245575560976225,-1.5647816211531764)-- (1.5382858178056378,-1.932505302124173);
\draw (1.5382858178056378,-1.932505302124173)-- (1.8309960746350507,-2.3002289830951694);
\draw [color=cqcqcq] (0.8310268655193396,-3.6459946901057236)-- (0.9362538054667601,-4.10766070864606);
\draw (0.9362538054667601,-4.10766070864606)-- (1.0414807454141806,-4.569326727186396);
\draw (-0.44650023202190303,-1.9495223883824437)-- (-0.5514277865470503,-2.4076601496523176);
\draw (-0.5514277865470503,-2.4076601496523176)-- (-0.6563553410721974,-2.8657979109221916);
\draw [color=cqcqcq] (-2.331979357502764,-2.923312369158491)-- (-2.627261089047286,-3.2934703362234354);
\draw (-2.627261089047286,-3.2934703362234354)-- (-2.922542820591808,-3.66362830328838);
\draw (-1.8024566839988228,-0.8666890459143738)-- (-2.226034004738546,-1.0703609717042517);
\draw [color=cqcqcq] (-2.226034004738546,-1.0703609717042517)-- (-2.6496113254782694,-1.2740328974941295);
\draw (-3.739503032689509,-0.0)-- (-4.213009338418993,-0.0);
\draw (-4.213009338418993,-0.0)-- (-4.686515644148477,-0.0);
\draw [color=cqcqcq] (-1.801548101210796,0.8685760985796095)-- (-2.224911904995333,1.0726914817458177);
\draw (-2.224911904995333,1.0726914817458177)-- (-2.64827570877987,1.2768068649120259);
\draw [color=cqcqcq] (-2.9187046762420463,3.6666867736051882)-- (1.0462651612432865,4.568233585886012);
\draw (1.0462651612432865,4.568233585886012)-- (4.223620723730235,2.0308758861449157);
\draw (4.223620723730235,2.0308758861449157)-- (4.221491680005189,-2.0352977370633947);
\draw (4.221491680005189,-2.0352977370633947)-- (1.0414807454141806,-4.569326727186396);
\draw (1.0414807454141806,-4.569326727186396)-- (-2.922542820591808,-3.66362830328838);
\draw (-2.922542820591808,-3.66362830328838)-- (-4.686515644148477,-0.0);
\draw [color=cqcqcq] (-4.686515644148477,-0.0)-- (-2.9187046762420463,3.6666867736051882);
\draw (-3.739503032689509,-0.0)-- (-2.64827570877987,1.2768068649120259);
\draw (-2.64827570877987,1.2768068649120259)-- (-2.328916793857265,2.9257528088995546);
\draw (-2.328916793857265,2.9257528088995546)-- (-0.6533539251791053,2.8664836731530596);
\draw [color=cqcqcq] (-0.6533539251791053,2.8664836731530596)-- (0.8348444858712378,3.645122441826121);
\draw (-0.44445845250279276,1.9499888933014011)-- (1.247213513195398,1.5634763997268373);
\draw [color=cqcqcq] (0.8348444858712378,3.645122441826121)-- (1.8334038643972352,2.2983103075984506);
\draw [color=cqcqcq] (1.8334038643972352,2.2983103075984506)-- (3.370146118052537,1.6204931577977892);
\draw [color=cqcqcq] (3.370146118052537,1.6204931577977892)-- (2.94,0.0);
\draw (1.247213513195398,1.5634763997268373)-- (2.0,0.0);
\draw (2.0,0.0)-- (1.245575560976225,-1.5647816211531764);
\draw (1.245575560976225,-1.5647816211531764)-- (-0.44650023202190303,-1.9495223883824437);
\draw (-0.44650023202190303,-1.9495223883824437)-- (-1.8024566839988228,-0.8666890459143738);
\draw [color=cqcqcq] (-1.8024566839988228,-0.8666890459143738)-- (-1.801548101210796,0.8685760985796095);
\draw [color=cqcqcq] (-1.801548101210796,0.8685760985796095)-- (-0.44445845250279276,1.9499888933014011);
\draw [color=cqcqcq] (-3.739503032689509,-0.0)-- (-2.6496113254782694,-1.2740328974941295);
\draw [color=cqcqcq] (-2.6496113254782694,-1.2740328974941295)-- (-2.331979357502764,-2.923312369158491);
\draw [color=cqcqcq] (-2.331979357502764,-2.923312369158491)-- (-0.6563553410721974,-2.8657979109221916);
\draw [color=cqcqcq] (-0.6563553410721974,-2.8657979109221916)-- (0.8310268655193396,-3.6459946901057236);
\draw [color=cqcqcq] (0.8310268655193396,-3.6459946901057236)-- (1.8309960746350507,-2.3002289830951694);
\draw [color=cqcqcq] (1.8309960746350507,-2.3002289830951694)-- (3.3684472940068986,-1.6240214773800357);
\draw [color=cqcqcq] (3.3684472940068986,-1.6240214773800357)-- (2.94,0.0);
\begin{scriptsize}
\draw [fill=black] (2.0,0.0) circle (1.5pt);
\draw [fill=black] (1.247213513195398,1.5634763997268373) circle (1.5pt);
\draw [fill=black] (-0.44445845250279276,1.9499888933014011) circle (1.5pt);
\draw [color=cqcqcq][fill=cqcqcq] (-1.801548101210796,0.8685760985796095) circle (1.5pt);
\draw [fill=black] (-1.8024566839988228,-0.8666890459143738) circle (1.5pt);
\draw [fill=black] (-0.44650023202190303,-1.9495223883824437) circle (1.5pt);
\draw [fill=black] (1.245575560976225,-1.5647816211531764) circle (1.5pt);
\draw [fill=black] (2.94,0.0) circle (1.5pt);
\draw [fill=black] (1.8334038643972352,2.2983103075984506) circle (1.5pt);
\draw [fill=black] (-0.6533539251791053,2.8664836731530596) circle (1.5pt);
\draw [fill=black] (-2.64827570877987,1.2768068649120259) circle (1.5pt);
\draw [color=cqcqcq][fill=cqcqcq] (-2.6496113254782694,-1.2740328974941295) circle (1.5pt);
\draw [fill=black] (-0.6563553410721974,-2.8657979109221916) circle (1.5pt);
\draw [fill=black] (1.8309960746350507,-2.3002289830951694) circle (1.5pt);
\draw [color=cqcqcq][fill=cqcqcq] (3.370146118052537,1.6204931577977892) circle (1.5pt);
\draw [color=cqcqcq][fill=cqcqcq] (0.8348444858712378,3.645122441826121) circle (1.5pt);
\draw [fill=black] (-2.328916793857265,2.9257528088995546) circle (1.5pt);
\draw [fill=black] (-3.739503032689509,-0.0) circle (1.5pt);
\draw [color=cqcqcq][fill=cqcqcq] (-2.331979357502764,-2.923312369158491) circle (1.5pt);
\draw [color=cqcqcq][fill=cqcqcq] (0.8310268655193396,-3.6459946901057236) circle (1.5pt);
\draw [color=cqcqcq][fill=cqcqcq] (3.3684472940068986,-1.6240214773800357) circle (1.5pt);
\draw [fill=black] (4.223620723730235,2.0308758861449157) circle (1.5pt);
\draw [fill=black] (1.0462651612432865,4.568233585886012) circle (1.5pt);
\draw [color=cqcqcq][fill=cqcqcq] (-2.9187046762420463,3.6666867736051882) circle (1.5pt);
\draw [fill=black] (-4.686515644148477,-0.0) circle (1.5pt);
\draw [fill=black] (-2.922542820591808,-3.66362830328838) circle (1.5pt);
\draw [fill=black] (1.0414807454141806,-4.569326727186396) circle (1.5pt);
\draw [fill=black] (4.221491680005189,-2.0352977370633947) circle (1.5pt);
\draw [fill=black] (-2.6238107350496556,3.2962197912523714) circle (1.5pt);
\draw [fill=black] (3.796883420891386,1.8256845219713524) circle (1.5pt);
\draw [fill=black] (0.9405548235572622,4.106678013856067) circle (1.5pt);
\draw [fill=black] (-4.213009338418993,-0.0) circle (1.5pt);
\draw [fill=black] (-2.627261089047286,-3.2934703362234354) circle (1.5pt);
\draw [fill=black] (0.9362538054667601,-4.10766070864606) circle (1.5pt);
\draw [fill=black] (3.794969487006044,-1.829659607221715) circle (1.5pt);
\draw [fill=black] (2.4699999999999998,0.0) circle (1.5pt);
\draw [fill=black] (1.5403086887963167,1.9308933536626438) circle (1.5pt);
\draw [fill=black] (-0.548906188840949,2.4082362832272306) circle (1.5pt);
\draw [fill=black] (-2.224911904995333,1.0726914817458177) circle (1.5pt);
\draw [fill=black] (-2.226034004738546,-1.0703609717042517) circle (1.5pt);
\draw [fill=black] (-0.5514277865470503,-2.4076601496523176) circle (1.5pt);
\draw [fill=black] (1.5382858178056378,-1.932505302124173) circle (1.5pt);
\end{scriptsize}
\end{tikzpicture}
\end{center}
\caption{A planar graph of girth $7$ on $42$ vertices that admits an induced forest on $34$ of its vertices, but no induced forest on $35$ or more of its vertices. \label{bolobolobo}}
\end{figure}

\section{Proof of Theorem~\ref{main}} \label{proofmain}

We first give a counter-example to the bound of Kowalik et al.~\cite{Kowalik}: we consider the disjoint union of $k$ cubes. There are $8k$ vertices and $12k$ edges, hence Kowalik et al.'s lower bound tells us that there is an induced forest of size at least $\frac{119(8k) -24(12k) -24}{128} = 5k + (k-1)\frac{3}{16}$. However there cannot be an induced forest of more than $5$ vertices in a cube (see Figure~\ref{bababa}), and thus the biggest induced forest in our graph contains $5k$ vertices, which contradicts the lower bound. Furthermore, by increasing $k$, we can see that the biggest induced forest can be arbitrarily smaller than the supposed lower bound.

The proofs of Theorems \ref{main} and \ref{bmain} follow the same scheme. They consist in looking for a minimal counter-example $G$, proving some structural properties on $G$ and concluding that it cannot verify Euler's formula, which is contradictory.

Consider $G = (V,E)$. For a set $S \subset V$, let $G - S$ be the graph constructed from $G$ by removing the vertices of $S$ and all the edges incident to some vertex of $S$. If $x \in V$, then we denote $G - \{x\}$ by $G - x$. For a set $S$ of vertices such that $S \cap V = \emptyset$, let $G + S$ be the graph constructed from $G$ by adding the vertices of $S$. If $x \notin V$, then we denote $G + \{x\}$ by $G + x$. For a set $F$ of pairs of vertices of $G$ such that $F \cap E = \emptyset$, let $G + F$ be the graph constructed from $G$ by adding the edges of $F$. If $e$ is a pair of vertices of $G$ and $e \notin E$, we denote $G + \{e\}$ by $G + e$. For a set $W \subset V$, we denote by $G[W]$ the subgraph of $G$ induced by $W$.

We call a vertex of degree $d$, at least $d$ and at most $d$, a \emph{$d$-vertex}, a \emph{$d^+$-vertex} and a \emph{$d^-$-vertex} respectively. Similarly, we call a cycle of length $l$, at least $l$ and at most $l$ a \emph{$l$-cycle}, a \emph{$l^+$-cycle} and a \emph{$l^-$-cycle} respectively, and by extension a face of length $l$, at least $l$ and at most $l$ a \emph{$l$-face}, a \emph{$l^+$-face} and a \emph{$l^-$-face} respectively.

Let ${\cal P}_4$ be the class of triangle-free planar graphs, and ${\cal P}_5$ be the class of planar graphs of girth at least $5$.

We will prove of the following more general statement than Theorem \ref{main}: 

\begin{theo} \label{genmain}
If $a$ and $b$ are positive constants such that equations (\ref{a})--(\ref{8a12b}) are verified, then $a(G) \ge an-bm$ for all $G \in {\cal P}_4$.
\end{theo}

\begin{eqnarray}
0 \le a \le 1 \label{a}\\
0 \le b \label{b}\\
a - 6b \le 0 \label{a6b}\\
3a - 10b \le 1 \label{3a10b}\\
8a - 12b \le 5 \label{8a12b}
\end{eqnarray}

\begin{figure}[h]
\definecolor{yqyqyq}{rgb}{0.4,0.4,0.4}
\begin{center}
\begin{tikzpicture}[line cap=round,line join=round,>=triangle 45,x=20.0cm,y=20.0cm]
\clip(0.0,0.74) rectangle (0.33,1.019);
\fill[color=yqyqyq,fill=yqyqyq,fill opacity=0.1] (-0.0,-0.0) -- (0.125,0.75) -- (0.1590909090909091,0.8636363636363636) -- (0.25,1.0) -- (1.0031103354657922,1.0) -- (1.0031103354657922,0.0) -- cycle;
\draw [domain=0.0:0.33] plot(\x,{(-1.0-0.0*\x)/-1.0});
\draw (1.0031103354657922,0.74) -- (1.0031103354657922,1.019);
\draw [domain=0.0:0.33] plot(\x,{(-0.0--6.0*\x)/1.0});
\draw [domain=0.0:0.33] plot(\x,{(--5.0--12.0*\x)/8.0});
\draw [domain=0.0:0.33] plot(\x,{(--1.0--10.0*\x)/3.0});
\draw (0.02422877058704006,1.018) node[anchor=north west] {\tiny{$a = 1$}};
\draw (0.2493432402421246,1.001516563000336) node[anchor=north west] {$(\frac{1}{4},1)$};
\draw (0.16435664995192747,0.8711994501738272) node[anchor=north west] {$(\frac{7}{44},\frac{8}{44})$};
\draw (0.128,0.78) node[anchor=north west] {$(\frac{1}{8},\frac{3}{4})$};
\draw (0.123,1.02) node[anchor=north west] {\tiny{$a = 6b$}};
\draw (0.17,1.02) node[anchor=north west] {\tiny{$3a - 10b = 1$}};
\draw (0.25,1.02) node[anchor=north west] {\tiny{$8a - 12b = 5$}};
\draw [->] (0.04942787894536187,0.7948541355859442) -- (0.04942787894536187,0.8291213103085657);
\draw [->] (0.04942787894536187,0.7948541355859442) -- (0.08369505366798262,0.7948541355859442);
\draw (0.03883047923594667,0.85) node[anchor=north west] {$a$};
\draw (0.062275055867725186,0.7949762332375673) node[anchor=north west] {$b$};
\begin{scriptsize}
\draw [fill=black] (0.125,0.75) circle (1.5pt);
\draw [fill=black] (0.1590909090909091,0.8636363636363636) circle (1.5pt);
\draw [fill=black] (0.25,1.0) circle (1.5pt);
\end{scriptsize}
\end{tikzpicture}
\end{center}
        \caption{\label{poly}The top-left part of the polygon of the constraints on $a$ and $b$.}
\end{figure}
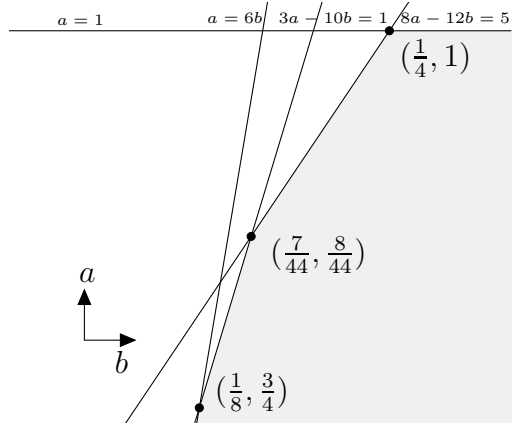

This series of inequalities defines a polygon represented in Figure~\ref{poly}, and for a triangle-free planar graph of given order $n$ and size $m$, the highest lower bound will be given by maximizing $an-bm$ for $a$ and $b$ in this polygon. This maximum will be achieved at a vertex of the polygon. Moreover, by Euler's formula, every triangle-free planar graph of order $n \ge 3$ and size $m$ satisfies $0 \le m \le 2n-4$. Therefore for $n \ge 3$ the maximum will always be achieved at the intersection of either $3a - 10b = 1$ and $8a - 12b = 5$, or $8a - 12b = 5$ and $a = 1$. The corresponding intersections are $(b,a) = (\frac{7}{44},\frac{38}{44})$ and $(b,a) = (\frac{1}{4},1)$, represented in Figure~\ref{poly}.

Let us show that any of the two lower bounds can be higher than the other, for graphs of arbitrarily high order. 

For the disjoint union of $k$ cubes (which is a graph of order $8k$ and size $12k$), the two lower bounds are equal to $5k$.

We consider now a graph composed of $k$ disjoint cubes, where we remove an edge from each cube. This graph has $8k$ vertices and $11k$ edges. In this case we have $n-\frac{m}4 = \frac{21}{4}k > \frac{38n-7m}{44} = \frac{227}{44}k$. More simply, for an independent set, $n-\frac{m}4 = n > \frac{38n-7m}{44} = \frac{38n}{44}$.

We now consider a graph composed of $k$ disjoint cubes, where we add an edge from each cube to the next one and an edge from the last one to the first one. This graph has $8k$ vertices and $13k$ edges. In this case, we have $n-\frac{m}4 = \frac{19}{4}k < \frac{38n-7m}{44} = \frac{213}{44}k$. For a quadrangulation on $n$ vertices and $2n-4$ edges (i.e. a planar graph on $n$ vertices that has only $4$-faces), $n-\frac{m}4 = \frac{n}{2} + 1 < \frac{38n-7m}{44} = \frac{6n + 7}{11}$.
\bigskip

Let us now proceed to the proof of Theorem~\ref{genmain}. For this proof we mainly adapt the methods of Kowalik et al. \cite{Kowalik}.

Let $G = (V,E)$ be a counter-example to Theorem~\ref{genmain} with the minimum order. Let $n = |V|$ and $m = |E|$. We will use the scheme presented in Observation~\ref{abg} for most of our lemmas.

\begin{obs} \label{abg}

Let $\alpha$, $\beta$, $\gamma$ be integers satisfying $\alpha \ge 1$, $\beta \ge 0$, $\gamma \ge 0$ and $a\alpha-b\beta \le \gamma$.

Let $H^* \in {\cal P}_4$ be a graph with $|V(H^*)| = n - \alpha$ and $|E(H^*)| \le m - \beta$.

By minimality of $G$, $H^*$ admits an induced forest of order at least $a(n-\alpha) - b(m - \beta)$.

For all induced forest $F^*$ of $H^*$ of order at least $a(n-\alpha) - b(m - \beta)$, if there is an induced forest $F$ of $G$ of order at least $|V(F^*)| + \gamma$, then we get a contradiction: as $a\alpha - b\beta \le \gamma$, we have $|V(F)| \ge an - bm$.
\end{obs}

Table~\ref{abgtab} contains the values of $(\alpha, \beta, \gamma)$ that will be used throughout this section. For each one, the inequality $a\alpha - b\beta \le \gamma$ is a consequence of the constraints (\ref{a})--(\ref{8a12b}).

\begin{table}[h]
\begin{center}
\begin{tabular}{|l|l|l|l|}
        \hline
        $\alpha$ & $\beta$ & $\gamma$ & proof \\
        \hline
        1 & 6 & 0 & $(\ref{a6b})$\\
        2 & 5 & 1 & $((\ref{a}) + (\ref{3a10b}))/2$\\
        3 & 5 & 2 & $(3(\ref{a}) + (\ref{3a10b}))/2$\\
        1 & 1 & 1 & $(\ref{a}) + (\ref{b})$\\
        5 & 9 & 3 & $((\ref{a}) + (\ref{a6b}) + (\ref{8a12b}))/2$\\
        6 & 8 & 4 & $((\ref{a}) + (\ref{8a12b}))*2/3$\\
        4 & 10 & 2 & $(\ref{a}) + (\ref{3a10b})$\\
        7 & 13 & 4 & $((\ref{a}) + 3(\ref{3a10b}) + 4(\ref{8a12b}))/6$\\
        3 & 10 & 1 & $(\ref{3a10b})$\\
        8 & 12 & 5 & $(\ref{8a12b})$\\
        6 & 14 & 3 & $((\ref{a6b}) + (\ref{3a10b}) + (\ref{8a12b}))/2$\\
        8 & 19 & 4 & $((\ref{a}) + (\ref{a6b}) + 2(\ref{3a10b}) + (\ref{8a12b}))/2$\\
        9 & 24 & 4 & $((\ref{a6b}) + 3(\ref{3a10b}) + (\ref{8a12b}))/2$\\
        10 & 23 & 5 & $((\ref{a}) + 9(\ref{3a10b}) + 4(\ref{8a12b}))/6$\\
        9 & 19 & 5 & $(3(\ref{a}) + (\ref{a6b}) + 2(\ref{3a10b}) + (\ref{8a12b}))/2$\\
        \hline
\end{tabular}
   \caption{\label{abgtab} The various triples ($\alpha$,$\beta$,$\gamma$) and the combinations of inequalities which imply $a\alpha - b\beta \le \gamma$.}
\end{center}
\end{table}
\bigskip

We will now prove a series of lemmas on the structure of $G$.

\begin{lemm} \label{2co}
Graph $G$ is 2-edge-connected.
\end{lemm}

\begin{proof}
By contradiction, suppose $V(G)$ is partitioned into two partite sets $V_1$ and $V_2$ such that there is at most one edge between vertices of $V_1$ and $V_2$. Consider graph $G[V_i]$ induced by the vertices of $V_i$ (for $i=1,2$) with $n_i=|V_i|$ vertices and $m_i=|E(G[V_i])|$ edges. By minimality of $G$, $G[V_i]$ admits an induced forest, say $F_i$, with at least $a n_i - b m_i$ vertices. Now the union of $F_1$ and $F_2$ (more formally, $G[V(F_1)\cup V(F_2)]$) is an induced forest of $G$ having at least $a n_1-b m_1 + a n_2-b m_2 = a(n_1+n_2)-b(m_1+m_2) \ge a n - b m$ vertices as $m\ge m_1+m_2$. A contradiction.
\end{proof}

In particular, Lemma~\ref{2co} implies that there is no $1^-$-vertex in $G$.

\begin{lemm} \label{degle5}
Every vertex in $G$ has degree at most $5$.
\end{lemm}

\begin{proof}
By contradiction, suppose $v \in V(G)$ is a $6^+$-vertex. Observation~\ref{abg} applied to $H^* = G - v$ with $(\alpha,\beta,\gamma) = (1,6,0)$ and $F = F^*$ completes the proof. 
\end{proof}

\begin{lemm} \label{deg3deg4}
If $v$ is a $3$-vertex adjacent to a $4^+$-vertex $w$ in $G$, then the two other neighbors of $v$ have a common neighbor different from $v$.
\end{lemm}

\begin{proof}
Let $x$ and $y$ be the two neighbors of $v$ different from $w$. Suppose that they do not have a common neighbor different from $v$. Let $H^* = G + xy - \{w,v\}$. Graph $H^*$ has $n-2$ vertices and $m' \le m-5$ edges. As $x$ and $y$ do not have a common neighbor in $G$ other than $v$, the addition of the edge $xy$ does not create any triangle in $H^*$, thus $H^* \in {\cal P}_4$. Let $F'$ be any induced forest of $H^*$. Adding $v$ to $F'$ (more formally, consider $G[V(F')\cup \{v\}]$) leads to an induced forest of $G$. Observation~\ref{abg} applied to $(\alpha,\beta,\gamma) = (2,5,1)$ completes the proof.
\end{proof}

\begin{lemm} \label{deg2deg4}
There is no 2-vertex adjacent to a $4^+$-vertex in $G$.
\end{lemm}

\begin{proof}
Let $v$ be a 2-vertex adjacent to a $4^+$-vertex $w$ and $H^*=G - \{v,w\}$. Graph $H^*$ has $n-2$ vertices and $m'\le m-5$ edges. Let $F'$ be any induced forest of $H^*$. Adding $v$ to $F'$ leads to an induced forest of $G$. Observation~\ref{abg} applied to $(\alpha,\beta,\gamma)=(2,5,1)$ completes the proof.
\end{proof}

\begin{lemm} \label{deg2deg3deg2}
There is no 3-vertex adjacent to two 2-vertices in $G$.
\end{lemm}

\begin{proof}
Let $v$ be a $3$-vertex adjacent to two 2-vertices $u$ and $w$ and $H^* = G - \{u,v,w\}$.  Graph $H^*$ has $n-3$ vertices and $m' = m-5$ edges. Let $F'$ be any induced forest of $H^*$. Adding $u$ and $w$ to $F'$ leads to an induced forest of $G$. Observation~\ref{abg} applied to $(\alpha,\beta,\gamma)=(3,5,2)$ completes the proof.
\end{proof}

\begin{lemm} \label{degge3}
Every vertex in $G$ has degree at least $3$.
\end{lemm}

\begin{proof}
Let $v$ be a $2$-vertex. 

Suppose that $v$ has a neighbor $u$ of degree $2$ and a neighbor $w$ of degree $3$. Let $H^* = G - \{u,v,w\}$.  Graph $H^*$ has $n-3$ vertices and $m' = m-5$ edges. Let $F'$ be any induced forest of $H^*$. Adding $u$ and $v$ to $F'$ leads to an induced forest of $G$. Observation~\ref{abg} applied to $(\alpha,\beta,\gamma)=(3,5,2)$ leads to a contradiction.

Suppose that $v$ has two neighbors of degree $3$, say $u$ and $w$. Consider three cases according to the number of neighbors $u$ and $w$ have in common.

\begin{itemize}
        \item Suppose $u$ and $w$ have only $v$ in common. Let $H^* = G + uw - v$. Graph $H^*$ has $n-1$ vertices and $m' = m-1$ edges. Observe that $H^* \in {\cal P}_4$. Let $F'$ be any induced forest of $H^*$. Adding $v$ to $F'$ (more formally, consider $G[V(F') \cup \{v\}]$) does not create any cycle (the edge $uw$ is just subdivided in $uv$, $vw$). Observation~\ref{abg} applied to $(\alpha,\beta,\gamma) = (1,1,1)$ leads to a contradiction.

        \item Suppose $u$ and $w$ have two neighbors in common, say $v$ and $x$. Let $y$ be the last neighbor of $u$. By Lemma~\ref{deg2deg3deg2}, both $x$ and $y$ have degree at least $3$.  Note that $x$ and $y$ are not adjacent because $G$ has girth at least $4$. Let $H^* = G - \{u,v,w,x,y\}$. Graph $H^*$ has $n-5$ vertices and, since $y$ and $w$ are not adjacent (otherwise $u$ and $w$ have three common neighbors), $m' \le m - 9$ edges. Let $F'$ be any induced forest of $H^*$. Adding $u$, $v$ and $w$ to $F'$ leads to an induced forest of $G$. Observation~\ref{abg} applied to $(\alpha,\beta,\gamma)=(5,9,3)$ leads to a contradiction.

        \item Suppose $u$ and $w$ have three neighbors in common. Let $x$ and $y$ be the ones that are not $v$. Suppose $x$ is a $4^+$-vertex and let $H^* = G - \{u,v,w,x,y\}$. Graph $H^*$ has $n-5$ vertices and $m' \le m - 9$ edges (recall that $y$ is a $3^+$-vertex by Lemma~\ref{deg2deg3deg2}). Let $F'$ be any induced forest of $H^*$. Adding $u$, $v$ and $w$ to $F'$ leads to an induced forest of $G$. Observation~\ref{abg} applied to $(\alpha,\beta,\gamma)=(5,9,3)$ leads to a contradiction.

        W.l.o.g. we assume that $x$ and $y$ are $3$-vertices. Let $z$ be the third neighbor of $x$. Let $H^* = G - \{u,v,w,x,y,z\}$. Graph $H^*$ has $n-6$ vertices and $m' \le m - 8$ edges. Let $F'$ be any induced forest of $H^*$. Adding $u$, $v$, $x$ and $y$ to $F'$ leads to an induced forest of $G$. Observation~\ref{abg} applied to $(\alpha,\beta,\gamma)=(6,8,4)$ leads to a contradiction. 
\end{itemize}

Therefore, by Lemmas~\ref{2co} and~\ref{deg2deg4}, every $2$-vertex has only neighbors of degree $2$. As $G$ is connected (Lemma~\ref{2co}), either $G$ does not have any $2$-vertex or it is 2-regular. If $G$ is $2$-regular, then $G$ is a $n$-cycle and thus $m = n$. Since $G \in {\cal P}_4$, we have $n \ge 4$. It is clear that $G$ has an induced forest of size $n-1$. Recall that $8a - 12b \le 5$ and $a \le 1$; this gives that $4(a-b) \le 3$. Since $n \ge 4$, we can deduce that $an-bm = (a-b)n \le n-1$. This contradicts the fact that $G$ is a counter-example. Therefore, $G$ has minimum degree at least $3$. This completes the proof.
\end{proof}

\begin{lemm} \label{4cycle}
There is no $4$-cycle in $G$ with
\begin{itemize}
\item at least one $4^+$-vertex and two opposite $3$-vertices
\item or one $3$-vertex opposite to a $4$-vertex that has an edge going to the interior of the cycle and one going to the exterior of it.
\end{itemize}

In particular there is no $4$-cycle with exactly three $3$-vertices in $G$. 
\end{lemm}

\begin{proof}
\begin{itemize}
        \item Let $C = v_0v_1v_2v_3$ be a cycle such that $v_0$ and $v_2$ have degree $3$ and $v_3$ is a $4^+$-vertex. Suppose $v_1$ is a $4^+$-vertex. Let $H^* = G - C$. Graph $H^*$ has $n-4$ vertices and $m' \le m - 10$ edges. Let $F'$ be any induced forest of $H^*$. Adding $v_0$ and $v_2$ to $F'$ leads to an induced forest of $G$. Observation~\ref{abg} applied to $(\alpha,\beta,\gamma)=(4,10,2)$ leads to a contradiction. Therefore $v_1$ has degree $3$.

 Let $u_0$, $u_1$ and $u_2$ be the third neighbors of $v_0$, $v_1$, and $v_2$, respectively. Suppose $u_0 = u_2$. Let $H^* = G - \{v_0,v_1,v_2,v_3,u_0\}$. Graph $H^*$ has $n-5$ vertices and $m' \le m - 9$ edges. Let $F'$ be any induced forest of $H^*$. Adding $v_0$, $v_1$ and $v_2$ to $F'$ leads to an induced forest of $G$. Observation~\ref{abg} applied to $(\alpha,\beta,\gamma)=(5,9,3)$ leads to a contradiction. So $u_0$ and $u_2$ are distinct.

        By Lemma~\ref{deg3deg4}, $u_0u_1 \in E$ and $u_1u_2 \in E$. Assume $u_0$ (or $u_2$) has at most one neighbor $w \notin \{v_0,v_1,v_2,v_3,u_0,u_1,u_2\}$. Let $H^* = G - \{v_0,v_1,v_2,v_3,u_0,u_1,u_2\}$. Graph $H^*$ has $n-7$ vertices and $m' \le m - 13$ edges. Let $F'$ be any induced forest of $H^*$. Adding $v_0$, $v_1$, $v_2$ and $u_0$ to $H^*$ leads to an induced forest of $G$. Observation~\ref{abg} applied to $(\alpha,\beta,\gamma)=(7,13,4)$ leads to a contradiction. Thus both of the vertices $u_0$ and $u_2$ have at least two neighbors that are not in $\{v_0,v_1,v_2,v_3,u_0,u_1,u_2\}$. Let $H^* = G - \{v_0,v_1,v_2,v_3,u_0,u_2\}$. Graph $H^*$ has $n-6$ vertices and $m' \le m-14$ edges. Let $F'$ be any induced forest of $H^*$. Adding the vertices $v_0$, $v_1$ and $v_2$ to $F'$ leads to an induced forest of $G$. Observation~\ref{abg} applied to $(\alpha,\beta,\gamma)=(6,14,3)$ leads to a contradiction.
        
        \item Let $C = v_0v_1v_2v_3$ be a cycle such that $v_0$ is a 3-vertex and $v_2$ is a 4-vertex with an edge going to the interior of the cycle and one going to the exterior of it. If $v_1$ and $v_3$ have degree $3$, then we fall into the previous case. Therefore w.l.o.g. $v_1$ is a $4^+$-vertex. Let $H^* = G - C$. Graph $H^*$ has $n - 4$ vertices and $m' \le m - 10$ edges. Let $F'$ be any induced forest of $H^*$. Adding $v_0$ and $v_2$ to $F'$ leads to an induced forest of $G$. Indeed, if adding $v_2$ creates a cycle, then there is a path from the interior to the exterior of $C$ in $H^*$, which is impossible. Observation~\ref{abg} applied to $(\alpha,\beta,\gamma)=(4,10,2)$ completes the proof.
\end{itemize}
\end{proof}

\begin{lemm} \label{3333}
There is no $4$-face with four $3$-vertices in $G$.
\end{lemm}

\begin{proof}
Suppose that there is such a $4$-face $C = v_0v_1v_2v_3$, and let $u_i$ be the third neighbor of $v_i$ for $i = 0..3$. In the following, we consider the indices of the $u_i$ and $v_i$ modulo $4$. If for some $i_0 \in \{0,1,2,3\}$, $u_{i_0} = u_{i_0+1}$, then we have a triangle. Suppose now that $u_{i_0} = u_{i_0+2}$ for some $i_0 \in \{0,1,2,3\}$, w.l.o.g. say $i_0 = 0$. In the cycle $v_{0}v_{1}v_{2}u_{0}$, the vertices $v_{0}$ and $v_{2}$ are two opposite $3$-vertices. By Lemma~\ref{4cycle}, $u_{0}$ is a 3-vertex. Observe that $u_1v_1$ and $u_3v_3$ are separated by the cycle $v_{0}v_{1}v_{2}u_{0}$. Hence one of them is a bridge, contradicting Lemma~\ref{2co}.

Therefore all the $u_i$ are distinct. We now consider the question of the presence or not of the edges $u_iu_{i+1}$. Consider the case $u_iu_{i+1} \notin E$ and $u_{i+1}u_{i+2} \notin E$ for some $i \in \{0,1,2,3\}$, w.l.o.g. say $i = 0$. If $u_0u_2 \in E$, then either $u_2u_3 \notin E$ or $u_0u_3 \notin E$ (otherwise $G$ has a triangle), and $u_1u_3 \notin E$ by planarity of $G$. Therefore up to the permutation of the indices, $u_0u_1 \notin E$, $u_1u_2 \notin E$ and $u_0u_2 \notin E$. We then define $H^* = G + x + \{xu_0, xu_1, xu_2\} - \{v_0,v_1,v_2,v_3\}$. Graph $H^*$ has $n-3$ vertices and $m' = m-5$ edges and belongs to ${\cal P}_4$ as $u_0u_1$, $u_0u_2$ and $u_1u_2$ are not in $E$. Let $F'$ be any induced forest of $H^*$. Let $F$ be the subgraph of $G$ induced by $V(F')\backslash\{x\}$ plus $v_0$, $v_1$ and $v_2$ if $x \in F'$ or plus $v_0$ and $v_2$ if $x \notin F'$. Subgraph $F$ is an induced forest of $G$. Hence, Observation~\ref{abg} applied to $(\alpha,\beta,\gamma) = (3,5,2)$ leads to a contradiction. Therefore there must be an $i$ such that $u_iu_{i+1} \in E$ and $u_{i+2}u_{i+3} \in E$, w.l.o.g. $u_0u_1 \in E$ and $u_2u_3 \in E$.

Let $G' = G - C$. Graph $G\rq{}$ has $n-4$ vertices and $m-8$ edges.

Let us now count, for each of the $u_i$'s, the number of the neighbors of $u_i$ that are not in $A = \{v_0,v_1,v_2,v_3,u_0,u_1,u_2,u_3\}$. The edges that are known in $G[A]$ are represented in Figure~\ref{GA}.

\begin{figure}[h]
\begin{center}
\begin{tikzpicture}[line cap=round,line join=round,>=triangle 45,x=1.0cm,y=1.0cm]
\clip(0.0,0.0) rectangle (6.7,6.7);
\draw (0.6879722606657073,6.03171952230884)-- (5.999242097886516,6.031719522308842);
\draw (5.999242097886509,0.7204496850880373)-- (0.6879722606657072,0.7204496850880311);
\draw (0.6879722606657073,6.03171952230884)-- (2.243607179276111,4.476084603698435);
\draw (2.243607179276111,4.476084603698435)-- (4.4436071792761105,4.476084603698435);
\draw (4.4436071792761105,4.476084603698435)-- (4.443607179276109,2.276084603698438);
\draw (4.443607179276109,2.276084603698438)-- (2.243607179276111,2.2760846036984357);
\draw (2.243607179276111,2.2760846036984357)-- (2.243607179276111,4.476084603698435);
\draw (2.243607179276111,2.2760846036984357)-- (0.6879722606657072,0.7204496850880311);
\draw (4.443607179276109,2.276084603698438)-- (5.999242097886509,0.7204496850880373);
\draw (4.4436071792761105,4.476084603698435)-- (5.999242097886516,6.031719522308842);
\begin{scriptsize}
\draw [fill=black] (2.243607179276111,4.476084603698435) circle (1.5pt);
\draw[color=black] (2.3667298238024124,4.707546961284052) node {$v_0$};
\draw [fill=black] (2.243607179276111,2.2760846036984357) circle (1.5pt);
\draw[color=black] (2.3667298238024124,2.016310029215045) node {$v_3$};
\draw [fill=black] (4.4436071792761105,4.476084603698435) circle (1.5pt);
\draw[color=black] (4.321748120224063,4.707546961284052) node {$v_1$};
\draw [fill=black] (4.443607179276109,2.276084603698438) circle (1.5pt);
\draw[color=black] (4.321748120224063,2.016310029215045) node {$v_2$};
\draw [fill=black] (0.6879722606657073,6.03171952230884) circle (1.5pt);
\draw[color=black] (0.8094356519546213,6.264841133131838) node {$u_0$};
\draw [fill=black] (5.999242097886516,6.031719522308842) circle (1.5pt);
\draw[color=black] (6.129042292071854,6.264841133131838) node {$u_1$};
\draw [fill=black] (5.999242097886509,0.7204496850880373) circle (1.5pt);
\draw[color=black] (6.129042292071854,0.4590158573672597) node {$u_2$};
\draw [fill=black] (0.6879722606657072,0.7204496850880311) circle (1.5pt);
\draw[color=black] (0.8094356519546213,0.4590158573672597) node {$u_3$};
\end{scriptsize}
\end{tikzpicture}
        \caption{\label{GA} The graph $G[A]$ (only the edges that are known to be there are represented).}
\end{center}
\end{figure}
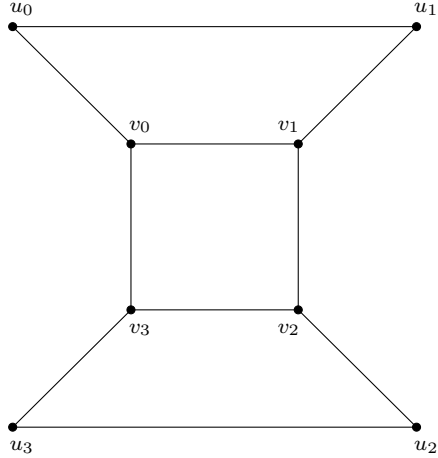

\begin{itemize}
        \item Suppose w.l.o.g. $u_0$ has only neighbors in $A$, and another $u_{i'}$ has at most one neighbor not in $A$. Let $H^* = G\rq{} - \{u_0,u_1,u_2,u_3\}$. Graph $H^*$ has $n-8$ vertices. By Lemma~\ref{degge3}, each of the $u_i$ has degree at least $3$. Graph $H^*$ has $m' \le m - 12$ edges. Let $F'$ be any induced forest of $H^*$. Adding the vertices $u_{0}$, $u_{i'}$, $v_1$, $v_2$ and $v_3$ to $F'$ leads to an induced forest of $G$. Observation~\ref{abg} applied to $(\alpha,\beta,\gamma) = (8,12,5)$ leads to a contradiction.

        \item Suppose w.l.o.g. $u_0$ has at most one neighbor not in $A$, and all the other $u_i$ have each at least one neighbor not in $A$. Vertex $u_0$ is not adjacent both to $u_2$ and $u_3$ since $G$ has girth at least $4$. Let $i_0$ be such that $i_0 \ne 0$ and $u_0u_{i_0} \notin E$ (either $i_0 = 2$ or $i_0 = 3$). Let $H^* = G\rq{} - \{u_{i_0+1},u_{i_0+2},u_{i_0+3}\}$ (we remove all the vertices of $A$ except $u_{i_0}$). Graph $H^*$ has $n-7$ vertices. Let us count the number of edges in $G\rq{}$ that have an endvertex in $\{u_{i_0+1},u_{i_0+2},u_{i_0+3}\}$. If $i_0 = 2$, then there are at least two edges for the neighbors of $u_1$ and $u_3$ that are not in $A$, plus the edges $u_0u_1$ and $u_2u_3$, plus one edge since $u_0$ has degree at least $3$, thus at least $5$ edges of $H^*$ have an endvertex in $\{u_{i_0+1},u_{i_0+2},u_{i_0+3}\}$.
If $i_0 = 3$, then there are at least two edges for the neighbors of $u_1$ and $u_2$ that are not in $A$, plus the edges $u_0u_1$ and $u_2u_3$, plus one edge since $u_0$ has degree at least $3$, thus at least $5$ edges of $H^*$ have an endvertex in $\{u_{i_0+1},u_{i_0+2},u_{i_0+3}\}$.
In both cases, $H^*$ has $m' \le m - 13$ edges. Let $F'$ be any induced forest of $H^*$. Adding the vertices $u_0$, $v_1$, $v_2$ and $v_3$ to $F'$ leads to an induced forest of $G$, since there is no path between $u_0$ and $u_{i_0}$ in $G[\{v_1,v_2,v_3,u_0,u_{i_0}\}]$. Observation~\ref{abg} applied to $(\alpha,\beta,\gamma) = (7,13,4)$ leads to a contradiction.

        \item So all the $u_i$ have at least two neighbors not in $A$. Let $H^* = G - \{v_0,v_1,v_2,v_3,u_0,u_2\}$. Graph $H^*$ has $n-6$ vertices and $m' \le m - 14$ edges, and if $F'$ is any induced forest in $H^*$, then adding the vertices $v_0$, $v_1$ and $v_2$ to $F'$ leads to an induced forest of $G$. Observation~\ref{abg} applied to $(\alpha,\beta,\gamma) = (6,14,3)$ leads to a contradiction and completes the proof.
\end{itemize}
\end{proof}

\begin{lemm} \label{sepcyclefirst}
There is no separating $4$-cycle with four $3$-vertices in $G$.
\end{lemm}

\begin{proof}
Let $C = v_0v_1v_2v_3$ be such a cycle. We will consider the indices of the $v_i$ modulo $4$ in what follows.
Since $G$ is $2$-edge-connected (Lemma~\ref{2co}), two of the $v_i$ have their third neighbor in the interior of $C$, and the two other have theirs outside of it. There is a $v_i$ such that the third neighbors of $v_{i+1}$ and $v_{i+2}$ are separated by $C$, w.l.o.g. for $i = 0$. Then let $u$ be the third neighbor of $v_0$. Let $H^* = G - C - u$. Graph $H^*$ has $n-5$ vertices, and $m' \le m-9$ edges. Let $F'$ be any induced forest of $H^*$. Adding the vertices $v_0$, $v_1$ and $v_2$ to $F'$ leads to a forest of $G$, thus Observation~\ref{abg} applied to  $(\alpha,\beta,\gamma) = (5,9,3)$ leads to a contradiction.
\end{proof}

\begin{lemm} \label{deg3deg5}
There is no $3$-vertex adjacent to a $5$-vertex in $G$.
\end{lemm}

\begin{proof}
Let $v$ be a $3$-vertex adjacent to a 5-vertex $u$. Let $w$ and $x$ be the two other neighbors of $v$. 

We first assume that $w$ or $x$, $w$ without loss of generality, is a $4^+$-vertex. Let $H^* = G - \{u,v,w\}$. Graph $H^*$ has $n-3$ vertices and $m' \le m-10$ edges. Let $F'$ be any induced forest of $H^*$. Adding $v$ to $F'$ leads to an induced forest of $G$. Thus Observation~\ref{abg} applied to $(\alpha,\beta,\gamma) = (3,10,1)$ leads to a contradiction.

Therefore $w$ and $x$ are $3$-vertices. By Lemma~\ref{deg3deg4}, $w$ and $x$ have a common neighbor (distinct from $v$), which has degree 3 by Lemma~\ref{4cycle}. Finally Lemmas~\ref{3333} and~\ref{sepcyclefirst} lead to a contradiction, completing the proof.
\end{proof}

\begin{lemm} \label{sepcycle}
There is no separating $4$-cycle with at least two $3$-vertices in $G$.
\end{lemm}

\begin{proof}
Let $C = v_0v_1v_2v_3$ be such a cycle.
By Lemmas~\ref{4cycle} and~\ref{sepcyclefirst}, $C$ has exactly two 3-vertices.
By Lemmas~\ref{degge3},~\ref{4cycle} and~\ref{deg3deg5}, the two 3-vertices are adjacent, the two other vertices have degree $4$ and none of the 4-vertices has a neighbor inside $C$ and the other one outside $C$. W.l.o.g. the 3-vertices are $v_0$ and $v_1$. Let $u_0$ and $u_1$ be the third neighbors of $v_0$ and $v_1$ respectively. 

If $u_0v_2 \in E$ or $u_1v_3 \in E$, say $u_0v_2 \in E$ w.l.o.g., then either $v_0v_1v_2u_0$ or $v_0v_3v_2u_0$ has a $3$-vertex ($v_0$) opposite to a $4$-vertex ($v_2$) with an edge going inside and one going outside of it, contradicting Lemma~\ref{4cycle}. Therefore $u_0v_2 \notin E$ and $u_1v_3 \notin E$.

By Lemma~\ref{deg3deg4}, $u_0u_1 \in E$; thus $C$ does not separate $u_0$ and $u_1$, say $u_0$ and $u_1$ are in the exterior of $C$ up to changing the plane embedding. By Lemmas~\ref{degge3}--\ref{deg3deg5}, $u_0$ and $u_1$ are $4$-vertices. At least one of $v_2$ or $v_3$, say $v_2$, has two neighbors inside of $C$ (otherwise the cycle is not separating). Let $H^* = G - \{v_0,v_1,v_3,u_1\}$. Graph $H^*$ has $n-4$ vertices and $m' \le m-10$ edges, and if $F'$ is any induced forest of $H^*$, then adding $v_0$ and $v_1$ to $F'$ leads to an induced forest of $G$ (since $v_2$ is only connected to the interior and $u_0$ to the exterior of $C$). Observation~\ref{abg} applied to $(\alpha, \beta, \gamma) = (4,10,2)$ completes the proof.
\end{proof}

\begin{lemm} \label{3344}
There is no $4$-face with exactly two $3$-vertices in $G$.
\end{lemm}

\begin{proof}

Let $C = v_0v_1v_2v_3$ be such a face. By Lemmas~\ref{degge3} and~\ref{4cycle} the two 3-vertices are adjacent. W.l.o.g. $v_0$ and $v_1$ have degree $3$, and $v_2$ and $v_3$ have degree $4$ (by Lemmas~\ref{degge3} and~\ref{deg3deg5}). Let $u_0$ and $u_1$ be the third neighbors of $v_0$ and $v_1$ respectively.
By Lemma~\ref{deg3deg4} applied to $v_0$ and $v_3$, and $v_1$ and $v_2$, $u_0u_1 \in E$. Then by Lemma~\ref{sepcycle}, $v_0v_1u_1u_0$ cannot be a separating cycle, and so it is the boundary of some 4-face. If both $u_0$ and $u_1$ have degree $3$, we have a contradiction by Lemma~\ref{3333}. If one has degree $3$ and the other has degree at least $4$, we have a contradiction by Lemma~\ref{4cycle}. Finally, by Lemma~\ref{deg3deg5}, $u_0$ and $u_1$ are $4$-vertices.

If $v_2$ is adjacent to $u_0$, then $u_0v_0v_1v_2$ is a separating $4$-cycle, with two 3-vertices, contradicting Lemma~\ref{sepcycle}. Hence $v_2u_0$ is not in $E$. Similarly, $v_3u_1$ is not in $E$. Since $G \in {\cal P}_4$, either $u_0$ and $v_2$ do not have a common neighbor, or $u_1$ and $v_3$ do not have a common neighbor. By symmetry assume that $u_0$ and $v_2$ do not have a common neighbor. Let $H^* = G + u_0v_2 - \{u_1,v_0,v_1,v_3\}$. Graph $H^*$ has $n-4$ vertices, $m' \le m - 10$ edges and belongs to ${\cal P}_4$. Let $F'$ be any induced forest of $H^*$. Adding $v_0$ and $v_1$ to $F'$ leads to an induced forest of $G$ (intuitively the edge $u_0v_2$ is just subdivided). Observation~\ref{abg} applied to $(\alpha,\beta,\gamma) = (4,10,2)$ completes the proof.
\end{proof}

\begin{lemm} \label{2vertices3}
There is no $4$-cycle with at least two $3$-vertices in $G$.
\end{lemm}

\begin{proof}
It follows from Lemmas~\ref{4cycle},~\ref{3333},~\ref{sepcycle} and~\ref{3344}.
\end{proof}

\begin{lemm} \label{3444}
There is no $4$-face with exactly one $3$-vertex in $G$.
\end{lemm}

\begin{proof}

Let $C = v_0v_1v_2v_3$ be such a face. W.l.o.g. $v_0$ is the 3-vertex and $v_1$, $v_2$ and $v_3$ are $4^+$-vertices. By Lemma~\ref{deg3deg5}, $v_1$ and $v_3$ are $4$-vertices. Let $u_0$ be the third neighbor of $v_0$. Vertex $u_0$ is different from $v_2$ and non-adjacent to $v_1$ and $v_3$ ($G$ is triangle-free). 

Let us first assume that $u_0v_2 \in E$. By Lemmas~\ref{degge3},~\ref{deg3deg5} and~\ref{2vertices3}, $u_0$ is a 4-vertex. Assume $v_2$ has degree $5$. Let $H^* = G - \{u_0,v_0,v_2\}$. Graph $H^*$ has $n-3$ vertices and $m-10$ edges. Let $F'$ be any induced forest of $H^*$. Adding the vertex $v_0$ to $F'$ leads to an induced forest of $G$. Observation~\ref{abg} applied to $(\alpha,\beta,\gamma) = (3,10,1)$ leads to a contradiction. Hence $v_2$ has degree 4. Then either $v_0v_1v_2u_0$ or $v_0v_3v_2u_0$ has a 3-vertex opposite to a 4-vertex with a neighbor in the interior and one in the exterior of it, contradicting Lemma~\ref{4cycle}.

Thus $u_0$ is non-adjacent to $v_2$. By Lemma~\ref{deg3deg4}, $v_1$ and $u_0$ have a common neighbor other than $v_0$, say $u_1$. It is distinct from all the vertices we defined previously. By Lemma~\ref{2vertices3} applied to $v_0v_1u_1u_0$, $u_0$ and $u_1$ have degree at least $4$. By Lemma~\ref{deg3deg5}, $u_0$ has degree exactly $4$.

Suppose $u_1v_3 \in E$. As $C$ is a face, the last neighbor of $v_1$ ($\ne v_0,v_2,u_1$), say $w_1$, is not in the interior of $C$. The cycle $v_0v_1u_1v_3$ separates $u_0$ and $v_2$.
Suppose first that $v_0v_1u_1v_3$ does not separate $u_0$ and $w_1$. Then $v_0v_1u_1u_0$ separates $v_3$ and $w_1$. Let $H^* = G - \{v_0,v_1,v_2,v_3,u_0,u_1\}$. Graph $G^*$ has $n-6$ vertices and $m' \le m-14$ edges. Let $F'$ be any induced forest of $H^*$. Adding the vertices $v_0$, $v_1$ and $v_3$ to $F'$ leads to an induced forest of $G$. Hence Observation~\ref{abg} applied to $(\alpha,\beta,\gamma) = (6,14,3)$ leads to a contradiction. Therefore $v_0v_1u_1v_3$ separates $u_0$ and $w_1$. Assume $u_1$ has degree $5$. Let $H^* = G - \{u_1,v_0,v_3\}$. Graph $H^*$ has $n-3$ vertices and $m-10$ edges. Let $F'$ be any induced forest of $H^*$. Adding the vertex $v_0$ to $F'$ leads to an induced forest of $G$. Observation~\ref{abg} applied to $(\alpha,\beta,\gamma) = (3,10,1)$ leads to a contradiction. Hence $u_1$ has degree 4. Then $v_0v_1u_1v_3$, $v_0u_0u_1v_3$ or $v_0v_1u_1u_0$ has a 3-vertex opposite to a 4-vertex with a neighbor in the interior and one in the exterior of it, contradicting Lemma~\ref{4cycle}. 

So $u_1$ cannot be adjacent to $v_3$. As $u_1v_3 \notin E$ and $u_0v_2 \notin E$, by Lemma~\ref{deg3deg4} $v_3$ and $u_0$ have a common neighbor distinct from $v_0$, say $u_3$. By what precedes and by symmetry, it is of degree at least $4$ and non-adjacent to $v_0$, $v_1$, $v_2$ and $u_1$ (it has a role similar to that of $u_1$, and is non-adjacent to $u_1$ because of the girth assumption). See Figure~\ref{preinterm} for a reminder of the structure of $G[\{v_0, v_1, v_2, v_3, u_0, u_1, u_3\}]$. Vertex $v_0$ has degree $3$, $v_1$, $v_3$ and $u_0$ are $4$-vertices, and $v_2$, $u_1$ and $u_3$ are $4^+$-vertices. Recall that $u_1v_3 \notin E$, $u_3v_1 \notin E$ and $u_0v_2 \notin E$.

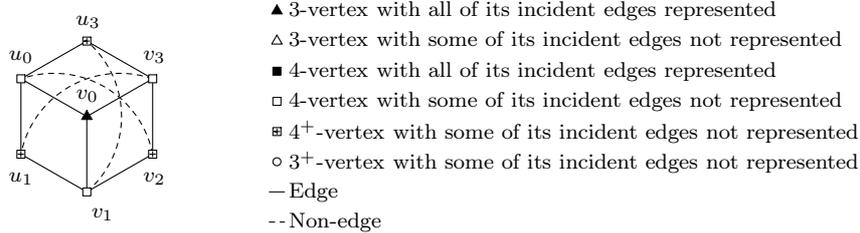
\begin{figure}[h]
\begin{center}
\begin{tikzpicture}[line cap=round,line join=round,>=triangle 45,x=1.0cm,y=1.0cm]
\clip(-2.0,-1.6) rectangle (11.0,1.6);
\draw (-0.8660254037844386,0.5)-- (0.0,-0.0);
\draw (0.0,-0.0)-- (0,-1.0);
\draw (0,-1.0)-- (0.8660254037844386,-0.4999999999999999);
\draw (0.8660254037844386,-0.4999999999999999)-- (0.8660254037844387,0.5);
\draw (0.8660254037844387,0.5)-- (0.0,1.0);
\draw (0.0,1.0)-- (-0.8660254037844386,0.5);
\draw (0.0,-0.0)-- (0.8660254037844387,0.5);
\draw (-0.8660254037844386,0.5)-- (-0.8660254037844386,-0.5);
\draw (-0.8660254037844386,-0.5)-- (0,-1.0);
\draw (2.4,-1)-- (2.6,-1);
\draw [dash pattern=on 2pt off 2pt](2.4,-1.4)-- (2.6,-1.4);
\draw [shift={(-0.8660254037844386,0.0)},dash pattern=on 2pt off 2pt]  plot[domain=-0.8570719478501312:0.8570719478501311,variable=\t]({1.0*1.3228756555322954*cos(\t r)+-0.0*1.3228756555322954*sin(\t r)},{0.0*1.3228756555322954*cos(\t r)+1.0*1.3228756555322954*sin(\t r)});
\draw [shift={(0.43301270189221935,-0.75)},dash pattern=on 2pt off 2pt]  plot[domain=1.2373231545430645:2.9514670502433265,variable=\t]({1.0*1.3228756555322954*cos(\t r)+-0.0*1.3228756555322954*sin(\t r)},{0.0*1.3228756555322954*cos(\t r)+1.0*1.3228756555322954*sin(\t r)});
\draw [shift={(-0.43301270189221924,-0.75)},dash pattern=on 2pt off 2pt]  plot[domain=0.19012560334646686:1.9042694990467288,variable=\t]({1.0*1.3228756555322951*cos(\t r)+-0.0*1.3228756555322951*sin(\t r)},{0.0*1.3228756555322951*cos(\t r)+1.0*1.3228756555322951*sin(\t r)});
\begin{scriptsize}
\draw [fill=black,shift={(0,0)}] (0,0) ++(0 pt,2.25pt) -- ++(1.9485571585149868pt,-3.375pt)--++(-3.8971143170299736pt,0 pt) -- ++(1.9485571585149868pt,3.375pt);
\draw[color=black] (0.008914485316013111,0.2735828720567711) node {$v_0$};
\draw [fill=white,shift={(0,1)}] (0,0) ++(1.5 pt,1.5pt) -- ++(-3pt,0pt)--++(0pt,-3 pt) -- ++(3pt,0pt) -- ++(0pt,3pt);
\draw [color=black] (0,1)-- ++(-1.0pt,0 pt) -- ++(2.0pt,0 pt) ++(-1.0pt,-1.0pt) -- ++(0 pt,2.0pt);
\draw[color=black] (0.01894019178265059,1.256102105787242) node {$u_3$};
\draw [fill=white,shift={(-0.8660254037844386,0.5)}] (0,0) ++(1.5 pt,1.5pt) -- ++(-3pt,0pt)--++(0pt,-3 pt) -- ++(3pt,0pt) -- ++(0pt,3pt);
\draw[color=black] (-0.8633219772814474,0.7949196083219188) node {$u_0$};
\draw [fill=white,shift={(-0.8660254037844386,-0.5)}] (0,0) ++(1.5 pt,1.5pt) -- ++(-3pt,0pt)--++(0pt,-3 pt) -- ++(3pt,0pt) -- ++(0pt,3pt);
\draw [color=black] (-0.8660254037844386,-0.5)-- ++(-1.0pt,0 pt) -- ++(2.0pt,0 pt) ++(-1.0pt,-1.0pt) -- ++(0 pt,2.0pt);
\draw[color=black] (-0.8633219772814474,-0.8184479388075878) node {$u_1$};
\draw [fill=white,shift={(0,-1.0)}] (0,0) ++(1.5 pt,1.5pt) -- ++(-3pt,0pt)--++(0pt,-3 pt) -- ++(3pt,0pt) -- ++(0pt,3pt);
\draw[color=black] (0.20937720171548772,-1.297332621394604) node {$v_1$};
\draw [fill=white,shift={(0.8660254037844386,-0.4999999999999999)}] (0,0) ++(1.5 pt,1.5pt) -- ++(-3pt,0pt)--++(0pt,-3 pt) -- ++(3pt,0pt) -- ++(0pt,3pt);
\draw [color=black] (0.8660254037844386,-0.4999999999999999)-- ++(-1.0pt,0 pt) -- ++(2.0pt,0 pt) ++(-1.0pt,-1.0pt) -- ++(0 pt,2.0pt);
\draw[color=black] (0.8911766543801111,-0.8184479388075878) node {$v_2$};
\draw [fill=white,shift={(0.8660254037844387,0.5)}] (0,0) ++(1.5 pt,1.5pt) -- ++(-3pt,0pt)--++(0pt,-3 pt) -- ++(3pt,0pt) -- ++(0pt,3pt);
\draw[color=black] (0.8911766543801111,0.7949196083219188) node {$v_3$};

\draw [fill=black,shift={(2.5,1.4)}] (0,0) ++(0 pt,2.25pt) -- ++(1.9485571585149868pt,-3.375pt)--++(-3.8971143170299736pt,0 pt) -- ++(1.9485571585149868pt,3.375pt);
\draw[color=black] (2.55,1.4) node[anchor = west] {$3$-vertex with all of its incident edges represented};
\draw [fill=white,shift={(2.5,1)}] (0,0) ++(0 pt,2.25pt) -- ++(1.9485571585149868pt,-3.375pt)--++(-3.8971143170299736pt,0 pt) -- ++(1.9485571585149868pt,3.375pt);
\draw[color=black] (2.55,1) node[anchor = west] {$3$-vertex with some of its incident edges not represented};
\draw [fill=black,shift={(2.5,0.6)}] (0,0) ++(1.5 pt,1.5pt) -- ++(-3pt,0pt)--++(0pt,-3 pt) -- ++(3pt,0pt) -- ++(0pt,3pt);
\draw[color=black] (2.55,0.6) node[anchor = west] {$4$-vertex with all of its incident edges represented};
\draw [fill=white,shift={(2.5,0.2)}] (0,0) ++(1.5 pt,1.5pt) -- ++(-3pt,0pt)--++(0pt,-3 pt) -- ++(3pt,0pt) -- ++(0pt,3pt);
\draw[color=black] (2.55,0.2) node[anchor = west] {$4$-vertex with some of its incident edges not represented};
\draw [fill=white,shift={(2.5,-0.2)}] (0,0) ++(1.5 pt,1.5pt) -- ++(-3pt,0pt)--++(0pt,-3 pt) -- ++(3pt,0pt) -- ++(0pt,3pt);
\draw [color=black] (2.5,-0.2)-- ++(-1.0pt,0 pt) -- ++(2.0pt,0 pt) ++(-1.0pt,-1.0pt) -- ++(0 pt,2.0pt);
\draw[color=black] (2.55,-0.2) node[anchor = west] {$4^+$-vertex with some of its incident edges not represented};
\draw [fill=white] (2.5,-0.6) circle (1.5pt);
\draw[color=black] (2.55,-0.6) node[anchor = west] {$3^+$-vertex with some of its incident edges not represented};
\draw[color=black] (2.55,-1) node[anchor = west] {Edge};
\draw[color=black] (2.55,-1.4) node[anchor = west] {Non-edge};
\end{scriptsize}
\end{tikzpicture}
\caption{\label{preinterm} Graph $G[\{v_0, v_1, v_2, v_3, u_0, u_1, u_3\}]$.}
\end{center}
\end{figure}

Let $w_0$, $w_1$ and $w_3$ be the fourth neighbors of $u_0$, $v_1$ and $v_3$ respectively. In the following we will no longer use the fact that $C$ is a face. By the girth assumption, $w_0$ is not adjacent to $u_1$ or $u_3$. Suppose $w_0$ is adjacent to $v_1$ or to $v_3$, say $w_0v_1 \in E$. Then by the girth assumption, $w_0v_2 \notin E$. By Lemma~\ref{2vertices3} applied to $v_0v_1w_0u_0$, $w_0$ is a $4^+$-vertex. Let $H^* = G - \{v_0,v_1,v_2,v_3,u_0,u_1,u_3,w_0\}$. Graph $H^*$ has $n-8$ vertices and $m' \le m-19$ edges. Let $F'$ be any induced forest of $H^*$. Adding the vertices $v_0$, $v_1$, $v_3$ and $u_0$ to $F'$ leads to an induced forest of $G$. Hence Observation~\ref{abg} applied to $(\alpha,\beta,\gamma) = (8,19,4)$ leads to a contradiction. So $w_0$ is not adjacent to $v_1$ or $v_3$. By symmetry, $w_0$, $w_1$ and $w_3$ are distinct.

Suppose $w_0v_2 \in E$. Assume that $C$ separates $w_1$ and $w_3$, or that it does not separate $w_1$ and $w_3$ nor $w_0$ and $w_1$. Then either $C$ or $v_0v_1v_2w_0u_0$ separates $w_1$ and $w_3$. Let $H^* = G - \{v_0,v_1,v_2,v_3,u_0,u_1,u_3,w_0\}$. Graph $H^*$ has $n-8$ vertices and $m' \le m-19$ edges. Let $F'$ be any induced forest of $H^*$. Adding the vertices $v_0$, $v_1$, $v_3$ and $u_0$ to $F'$ leads to an induced forest of $G$. Hence Observation~\ref{abg} applied to $(\alpha,\beta,\gamma) = (8,19,4)$ leads to a contradiction. Thus $C$ does not separate $w_1$ and $w_3$ but separates $w_1$ and $w_0$. Let $H^* = G - \{v_0,v_1,v_2,v_3,u_0,u_1,u_3,w_3\}$. Graph $H^*$ has $n-8$ vertices and $m' \le m-19$ edges. Let $F'$ be any induced forest of $H^*$. Adding the vertices $v_0$, $v_1$, $v_3$ and $u_0$ to $F'$ leads to an induced forest of $G$. Hence Observation~\ref{abg} applied to $(\alpha,\beta,\gamma) = (8,19,4)$ leads to a contradiction. So $w_0v_2 \notin E$, and similarly $w_1u_3 \notin E$ and $w_3u_1 \notin E$. 

Thus the only edges that may or may not exist between the vertices we defined are $w_0w_1$, $w_0w_3$ and $w_1w_3$. See Figure~\ref{interm} for a reminder of the edges and vertices we know to this point. Vertex $v_0$ has degree $3$, $v_1$, $v_3$ and $u_0$ are $4$-vertices and $v_2$, $u_1$ and $u_3$ are $4^+$-vertices. Vertices $v_0$, $v_1$, $v_3$ and $u_0$ have all their incident edges represented in Figure~\ref{interm}.

\begin{figure}[h]
\begin{center}
\begin{tikzpicture}[line cap=round,line join=round,>=triangle 45,x=1.0cm,y=1.0cm]
\clip(-2.0,-2.4) rectangle (2.1,1.6);
\draw (-0.8660254037844386,0.5)-- (0.0,-0.0);
\draw (0.0,-0.0)-- (0,-1.0);
\draw (0,-1.0)-- (0.8660254037844386,-0.4999999999999999);
\draw (0.8660254037844386,-0.4999999999999999)-- (0.8660254037844387,0.5);
\draw (0.8660254037844387,0.5)-- (0.0,1.0);
\draw (0.0,1.0)-- (-0.8660254037844386,0.5);
\draw (0.0,-0.0)-- (0.8660254037844387,0.5);
\draw (-0.8660254037844386,0.5)-- (-0.8660254037844386,-0.5);
\draw (-0.8660254037844386,-0.5)-- (0,-1.0);
\draw (-0.8660254037844386,0.5)-- (-1.7320508075688772,1.0);
\draw (0.8660254037844387,0.5)-- (1.7320508075688774,0.9999999999999997);
\draw (0,-1.0)-- (0,-2.0000000000000004);
\begin{scriptsize}
\draw [fill=black,shift={(0,0)}] (0,0) ++(0 pt,2.25pt) -- ++(1.9485571585149868pt,-3.375pt)--++(-3.8971143170299736pt,0 pt) -- ++(1.9485571585149868pt,3.375pt);
\draw[color=black] (0.008914485316013111,0.2735828720567711) node {$v_0$};
\draw [fill=white,shift={(0,1)}] (0,0) ++(1.5 pt,1.5pt) -- ++(-3pt,0pt)--++(0pt,-3 pt) -- ++(3pt,0pt) -- ++(0pt,3pt);
\draw [color=black] (0,1)-- ++(-1.0pt,0 pt) -- ++(2.0pt,0 pt) ++(-1.0pt,-1.0pt) -- ++(0 pt,2.0pt);
\draw[color=black] (0.01894019178265059,1.256102105787242) node {$u_3$};
\draw [fill=black,shift={(-0.8660254037844386,0.5)}] (0,0) ++(1.5 pt,1.5pt) -- ++(-3pt,0pt)--++(0pt,-3 pt) -- ++(3pt,0pt) -- ++(0pt,3pt);
\draw[color=black] (-0.8633219772814474,0.7949196083219188) node {$u_0$};
\draw [fill=white,shift={(-0.8660254037844386,-0.5)}] (0,0) ++(1.5 pt,1.5pt) -- ++(-3pt,0pt)--++(0pt,-3 pt) -- ++(3pt,0pt) -- ++(0pt,3pt);
\draw [color=black] (-0.8660254037844386,-0.5)-- ++(-1.0pt,0 pt) -- ++(2.0pt,0 pt) ++(-1.0pt,-1.0pt) -- ++(0 pt,2.0pt);
\draw[color=black] (-0.8633219772814474,-0.8184479388075878) node {$u_1$};
\draw [fill=black,shift={(0,-1.0)}] (0,0) ++(1.5 pt,1.5pt) -- ++(-3pt,0pt)--++(0pt,-3 pt) -- ++(3pt,0pt) -- ++(0pt,3pt);
\draw[color=black] (0.20937720171548772,-1.297332621394604) node {$v_1$};
\draw [fill=white,shift={(0.8660254037844386,-0.4999999999999999)}] (0,0) ++(1.5 pt,1.5pt) -- ++(-3pt,0pt)--++(0pt,-3 pt) -- ++(3pt,0pt) -- ++(0pt,3pt);
\draw [color=black] (0.8660254037844386,-0.4999999999999999)-- ++(-1.0pt,0 pt) -- ++(2.0pt,0 pt) ++(-1.0pt,-1.0pt) -- ++(0 pt,2.0pt);
\draw[color=black] (0.8911766543801111,-0.8184479388075878) node {$v_2$};
\draw [fill=black,shift={(0.8660254037844387,0.5)}] (0,0) ++(1.5 pt,1.5pt) -- ++(-3pt,0pt)--++(0pt,-3 pt) -- ++(3pt,0pt) -- ++(0pt,3pt);
\draw[color=black] (0.8911766543801111,0.7949196083219188) node {$v_3$};
\draw [fill=white] (-1.7320508075688772,1.0) circle (1.5pt);
\draw[color=black] (-1.6252756687458956,1.2059735734540544) node {$w_0$};
\draw [fill=white] (1.7320508075688774,0.9999999999999997) circle (1.5pt);
\draw[color=black] (1.873695888110584,1.1758964540541421) node {$w_3$};
\draw [fill=white] (0,-2.0000000000000004) circle (1.5pt);
\draw[color=black] (0.20937720171548772,-2.3) node {$w_1$};
\end{scriptsize}
\end{tikzpicture}
\caption{\label{interm}Vertices $v_0$, $v_1$, $v_2$, $v_3$, $u_0$, $u_1$, $u_3$, $w_0$, $w_1$ and $w_3$.}
\end{center}
\end{figure}
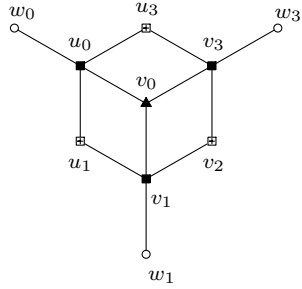

Suppose $w_0w_1 \notin E$, $w_0w_3 \notin E$, and $w_1w_3 \notin E$. Let $H^* = G + x + \{xw_0,xw_1,xw_3\} - \{v_0,v_1,v_2,v_3,u_0,u_1,u_3\}$. Graph $H^*$ has $n-6$ vertices and $m' \le m-14$ edges, and is in ${\cal P}_4$. Let $F'$ be any induced forest of $H^*$. Either $x \in F'$, then the graph induced by $V(F') \cup \{v_0,v_1,v_3,u_0\} \backslash \{x\}$ in $G$ is a forest, or $x \notin F'$, then adding $v_1$, $v_3$ and $u_0$ to $F'$ leads to an induced forest of $G$. Observation~\ref{abg} applied to $(\alpha,\beta,\gamma) = (6,14,3)$ leads to a contradiction. Thus there is at least one edge among $w_0w_1$, $w_0w_3$ and $w_1w_3$. Moreover, since there is no triangle in $G$, there are no more than two of these edges. W.l.o.g. let us assume that $w_0w_1 \notin E$ and $w_0w_3 \in E$. 

Let us now prove some claims that we will use later :

\begin{enumerate}[label=(\alph*)]

\item \label{ma}Suppose that $w_0$ and $w_1$ are $4^+$-vertices, or that one is a $3$-vertex, the other a $4^+$-vertex, and $v_2$, $u_1$ or $u_3$ has degree $5$. Let $H^* = G - \{v_0,v_1,v_2,v_3,u_0,u_1,u_3,w_0,w_1\}$. Graph $H^*$ has $n-9$ vertices and $m' \le m - 24$ edges, and adding $v_0$, $v_1$, $v_3$ and $u_0$ to any induced forest of $H^*$ leads to an induced forest of $G$. Observation~\ref{abg} applied to $(\alpha,\beta,\gamma) = (9,24,4)$ leads to a contradiction.

\item \label{mb} Suppose $w_0$ or $w_3$, say $w_{i_0}$, is a $3$-vertex and either one of the $w_i$ is a $4^+$-vertex, or $w_1w_3 \notin E$. Let $H^* = G - \{v_0,v_1,v_2,v_3,u_0,u_1,u_3,w_0,w_1,w_3\}$. Graph $H^*$ has $n-10$ vertices and $m' \le m - 23$ edges, and adding $v_0$, $v_1$, $v_3$, $u_0$ and $w_{i_0}$ to any induced forest of $H^*$ leads to an induced forest of $G$. Observation~\ref{abg} applied to $(\alpha,\beta,\gamma) = (10,23,5)$ leads to a contradiction.

\item \label{mc} Suppose $w_0$ and $w_3$ are $3$-vertices and $w_1$ and $w_3$ are adjacent. Let $H^* = G - \{v_0,v_1,v_3,u_0,u_1,u_3,w_0,w_1,w_3\}$. Graph $H^*$ has $n-9$ vertices and $m' \le m - 19$ edges, and adding $v_0$, $v_1$, $u_0$, $w_0$ and $w_3$ to any induced forest of $H^*$ leads to an induced forest of $G$ (by planarity, since $w_1w_3 \in E$ and $w_0w_3 \in E$, the cycle $v_0v_1w_1w_3v_3$ separates $v_2$ from $w_0$ in $G$). Observation~\ref{abg} applied to $(\alpha,\beta,\gamma) = (9,19,5)$ leads to a contradiction.
\end{enumerate}

If $w_1w_3 \in E$, then both $w_0$ and $w_3$ are $4^+$-vertices (by \ref{mb} and \ref{mc}), and by symmetry $w_1$ is also a $4^+$-vertex, which is impossible (by \ref{ma}). Hence $w_1w_3 \notin E$.

\begin{figure}[h]
\begin{center}
\begin{tikzpicture}[line cap=round,line join=round,>=triangle 45,x=1.0cm,y=1.0cm]
\clip(-2.0,-2.5) rectangle (2.0,1.7);
\draw (-0.8660254037844386,0.5)-- (0.0,-0.0);
\draw (0.0,-0.0)-- (0,-1.0);
\draw (0,-1.0)-- (0.8660254037844386,-0.4999999999999999);
\draw (0.8660254037844386,-0.4999999999999999)-- (0.8660254037844387,0.5);
\draw (0.8660254037844387,0.5)-- (0.0,1.0);
\draw (0.0,1.0)-- (-0.8660254037844386,0.5);
\draw (0.0,-0.0)-- (0.8660254037844387,0.5);
\draw (-0.8660254037844386,0.5)-- (-0.8660254037844386,-0.5);
\draw (-0.8660254037844386,-0.5)-- (0,-1.0);
\draw (-0.8660254037844386,0.5)-- (-1.7320508075688772,1.0);
\draw (0.8660254037844387,0.5)-- (1.7320508075688774,0.9999999999999997);
\draw (0,-1.0)-- (0,-2.0000000000000004);
\draw [shift={(0,-2.543640645068352)}] plot[domain=1.116167186841501:2.0254254667482923,variable=\t]({1.0*3.9442856161008986*cos(\t r)+-0.0*3.9442856161008986*sin(\t r)},{0.0*3.9442856161008986*cos(\t r)+1.0*3.9442856161008986*sin(\t r)});
\draw [shift={(-2.206307779361453,1.2738123903295004)},dash pattern=on 2pt off 2pt]  plot[domain=5.3054005551992285:6.213772507963347,variable=\t]({1.0*3.9478654465857654*cos(\t r)+-0.0*3.9478654465857654*sin(\t r)},{0.0*3.9478654465857654*cos(\t r)+1.0*3.9478654465857654*sin(\t r)});
\draw [shift={(2.2063077793614534,1.2738123903295004)},dash pattern=on 2pt off 2pt]  plot[domain=3.211005452806033:4.119377405570152,variable=\t]({1.0*3.9478654465857654*cos(\t r)+-0.0*3.9478654465857654*sin(\t r)},{0.0*3.9478654465857654*cos(\t r)+1.0*3.9478654465857654*sin(\t r)});

\begin{scriptsize}
\draw [fill=black,shift={(0,0)}] (0,0) ++(0 pt,2.25pt) -- ++(1.9485571585149868pt,-3.375pt)--++(-3.8971143170299736pt,0 pt) -- ++(1.9485571585149868pt,3.375pt);
\draw[color=black] (0.008914485316013111,0.2735828720567711) node {$v_0$};
\draw [fill=white,shift={(0,1)}] (0,0) ++(1.5 pt,1.5pt) -- ++(-3pt,0pt)--++(0pt,-3 pt) -- ++(3pt,0pt) -- ++(0pt,3pt);
\draw[color=black] (0.01894019178265059,1.256102105787242) node {$u_3$};
\draw [fill=black,shift={(-0.8660254037844386,0.5)}] (0,0) ++(1.5 pt,1.5pt) -- ++(-3pt,0pt)--++(0pt,-3 pt) -- ++(3pt,0pt) -- ++(0pt,3pt);
\draw[color=black] (-0.8633219772814474,0.7949196083219188) node {$u_0$};
\draw [fill=white,shift={(-0.8660254037844386,-0.5)}] (0,0) ++(1.5 pt,1.5pt) -- ++(-3pt,0pt)--++(0pt,-3 pt) -- ++(3pt,0pt) -- ++(0pt,3pt);
\draw[color=black] (-0.8633219772814474,-0.8184479388075878) node {$u_1$};
\draw [fill=black,shift={(0,-1.0)}] (0,0) ++(1.5 pt,1.5pt) -- ++(-3pt,0pt)--++(0pt,-3 pt) -- ++(3pt,0pt) -- ++(0pt,3pt);
\draw[color=black] (0.20937720171548772,-1.297332621394604) node {$v_1$};
\draw [fill=white,shift={(0.8660254037844386,-0.4999999999999999)}] (0,0) ++(1.5 pt,1.5pt) -- ++(-3pt,0pt)--++(0pt,-3 pt) -- ++(3pt,0pt) -- ++(0pt,3pt);
\draw[color=black] (0.8911766543801111,-0.8184479388075878) node {$v_2$};
\draw [fill=black,shift={(0.8660254037844387,0.5)}] (0,0) ++(1.5 pt,1.5pt) -- ++(-3pt,0pt)--++(0pt,-3 pt) -- ++(3pt,0pt) -- ++(0pt,3pt);
\draw[color=black] (0.8911766543801111,0.7949196083219188) node {$v_3$};
\draw [fill=white,shift={(-1.7320508075688772,1.0)}] (0,0) ++(1.5 pt,1.5pt) -- ++(-3pt,0pt)--++(0pt,-3 pt) -- ++(3pt,0pt) -- ++(0pt,3pt);
\draw [color=black] (-1.7320508075688772,1.0)-- ++(-1.0pt,0 pt) -- ++(2.0pt,0 pt) ++(-1.0pt,-1.0pt) -- ++(0 pt,2.0pt);
\draw[color=black] (-1.6252756687458956,1.2059735734540544) node {$w_0$};
\draw [fill=white,shift={(1.7320508075688774,0.9999999999999997)}] (0,0) ++(1.5 pt,1.5pt) -- ++(-3pt,0pt)--++(0pt,-3 pt) -- ++(3pt,0pt) -- ++(0pt,3pt);
\draw [color=black] (1.7320508075688774,0.9999999999999997)-- ++(-1.0pt,0 pt) -- ++(2.0pt,0 pt) ++(-1.0pt,-1.0pt) -- ++(0 pt,2.0pt);
\draw[color=black] (1.773695888110584,1.1758964540541421) node {$w_3$};
\draw [fill=white,shift={(0,-2.0000000000000004)}] (0,0) ++(0 pt,2.25pt) -- ++(1.9485571585149868pt,-3.375pt)--++(-3.8971143170299736pt,0 pt) -- ++(1.9485571585149868pt,3.375pt);
\draw[color=black] (0.20937720171548772,-2.3) node {$w_1$};
\end{scriptsize}
\end{tikzpicture}
\caption{\label{intermezmoi}Vertices $v_0$, $v_1$, $v_2$, $v_3$, $u_0$, $u_1$, $u_3$, $w_0$, $w_1$ and $w_3$.}
\end{center}
\end{figure}

Therefore $w_0$ and $w_3$ are $4^+$-vertices (by \ref{mb}), thus $w_1$ has degree $3$ (by \ref{ma}), and $v_2$, $u_1$ and $u_3$ have degree $4$ (by \ref{ma}) (see Figure~\ref{intermezmoi}). Let $y_0$ and $y_1$ the two neighbors of $w_1$ other than $v_1$. By Lemma~\ref{deg3deg4} they have a common neighbor other than $w_1$, say $t$. So by Lemmas~\ref{deg3deg5} and~\ref{2vertices3} in $w_1y_0ty_1$, $y_0$ and $y_1$ have degree $4$, and by Lemma~\ref{deg3deg4} each one is adjacent either to $v_2$ or to $u_1$. If they are both adjacent to the same one, say $v_2$ w.l.o.g., then either $v_2v_1w_1y_0$ or $v_2v_1w_1y_1$ is a $4$-cycle with a 3-vertex ($w_1$) opposite to a 4-vertex ($v_2$) that has both an edge going outside and one going inside of it, which is impossible by Lemma~\ref{4cycle}. W.l.o.g., say $y_0$ is adjacent to $v_2$ and $y_1$ is adjacent to $u_1$. At this point we know that $v_0$, $v_1$, $v_2$, $v_3$, $u_0$, $u_1$, $w_1$, $y_0$ and $y_1$ are distinct and do not share an edge that we do not already know. See Figure~\ref{intermbis} for a reminder of the edges and vertices we know to this point.

\begin{figure}[h]
\begin{center}
\begin{tikzpicture}[line cap=round,line join=round,>=triangle 45,x=1.0cm,y=1.0cm]
\clip(-2.0,-2.5) rectangle (2.0,1.7);
\draw (-0.8660254037844386,0.5)-- (0.0,-0.0);
\draw (0.0,-0.0)-- (0,-1.0);
\draw (0,-1.0)-- (0.8660254037844386,-0.4999999999999999);
\draw (0.8660254037844386,-0.4999999999999999)-- (0.8660254037844387,0.5);
\draw (0.8660254037844387,0.5)-- (0.0,1.0);
\draw (0.0,1.0)-- (-0.8660254037844386,0.5);
\draw (0.0,-0.0)-- (0.8660254037844387,0.5);
\draw (-0.8660254037844386,0.5)-- (-0.8660254037844386,-0.5);
\draw (-0.8660254037844386,-0.5)-- (0,-1.0);
\draw (-0.8660254037844386,0.5)-- (-1.7320508075688772,1.0);
\draw (0.8660254037844387,0.5)-- (1.7320508075688774,0.9999999999999997);
\draw (0,-1.0)-- (0,-2.0000000000000004);
\draw [shift={(0,-2.543640645068352)}] plot[domain=1.116167186841501:2.0254254667482923,variable=\t]({1.0*3.9442856161008986*cos(\t r)+-0.0*3.9442856161008986*sin(\t r)},{0.0*3.9442856161008986*cos(\t r)+1.0*3.9442856161008986*sin(\t r)});
\draw (-0.8660254037844386,-0.5)-- (-0.8660254037844386,-1.5000000000000002);
\draw (0,-2.0000000000000004)-- (-0.8660254037844386,-1.5000000000000002);
\draw (0,-2.0000000000000004)-- (0.866025403784439,-1.5);
\draw (0.866025403784439,-1.5)-- (0.8660254037844386,-0.4999999999999999);
\begin{scriptsize}
\draw [fill=black,shift={(0,0)}] (0,0) ++(0 pt,2.25pt) -- ++(1.9485571585149868pt,-3.375pt)--++(-3.8971143170299736pt,0 pt) -- ++(1.9485571585149868pt,3.375pt);
\draw[color=black] (0.008914485316013111,0.2735828720567711) node {$v_0$};
\draw [fill=white,shift={(0,1)}] (0,0) ++(1.5 pt,1.5pt) -- ++(-3pt,0pt)--++(0pt,-3 pt) -- ++(3pt,0pt) -- ++(0pt,3pt);
\draw[color=black] (0.01894019178265059,1.256102105787242) node {$u_3$};
\draw [fill=black,shift={(-0.8660254037844386,0.5)}] (0,0) ++(1.5 pt,1.5pt) -- ++(-3pt,0pt)--++(0pt,-3 pt) -- ++(3pt,0pt) -- ++(0pt,3pt);
\draw[color=black] (-0.8633219772814474,0.7949196083219188) node {$u_0$};
\draw [fill=white,shift={(-0.8660254037844386,-0.5)}] (0,0) ++(1.5 pt,1.5pt) -- ++(-3pt,0pt)--++(0pt,-3 pt) -- ++(3pt,0pt) -- ++(0pt,3pt);
\draw[color=black] (-1.1047540490234774,-0.6184479388075878) node {$u_1$};
\draw [fill=black,shift={(0,-1.0)}] (0,0) ++(1.5 pt,1.5pt) -- ++(-3pt,0pt)--++(0pt,-3 pt) -- ++(3pt,0pt) -- ++(0pt,3pt);
\draw[color=black] (0.20937720171548772,-1.297332621394604) node {$v_1$};
\draw [fill=white,shift={(0.8660254037844386,-0.4999999999999999)}] (0,0) ++(1.5 pt,1.5pt) -- ++(-3pt,0pt)--++(0pt,-3 pt) -- ++(3pt,0pt) -- ++(0pt,3pt);
\draw[color=black] (1.104702635808183,-0.6184479388075878) node {$v_2$};
\draw [fill=black,shift={(0.8660254037844387,0.5)}] (0,0) ++(1.5 pt,1.5pt) -- ++(-3pt,0pt)--++(0pt,-3 pt) -- ++(3pt,0pt) -- ++(0pt,3pt);
\draw[color=black] (0.8911766543801111,0.7949196083219188) node {$v_3$};
\draw [fill=white,shift={(-1.7320508075688772,1.0)}] (0,0) ++(1.5 pt,1.5pt) -- ++(-3pt,0pt)--++(0pt,-3 pt) -- ++(3pt,0pt) -- ++(0pt,3pt);
\draw [color=black] (-1.7320508075688772,1.0)-- ++(-1.0pt,0 pt) -- ++(2.0pt,0 pt) ++(-1.0pt,-1.0pt) -- ++(0 pt,2.0pt);
\draw[color=black] (-1.6252756687458956,1.2059735734540544) node {$w_0$};
\draw [fill=white,shift={(1.7320508075688774,0.9999999999999997)}] (0,0) ++(1.5 pt,1.5pt) -- ++(-3pt,0pt)--++(0pt,-3 pt) -- ++(3pt,0pt) -- ++(0pt,3pt);
\draw [color=black] (1.7320508075688774,0.9999999999999997)-- ++(-1.0pt,0 pt) -- ++(2.0pt,0 pt) ++(-1.0pt,-1.0pt) -- ++(0 pt,2.0pt);
\draw[color=black] (1.773695888110584,1.1758964540541421) node {$w_3$};
\draw [fill=black,shift={(0,-2.0000000000000004)}] (0,0) ++(0 pt,2.25pt) -- ++(1.9485571585149868pt,-3.375pt)--++(-3.8971143170299736pt,0 pt) -- ++(1.9485571585149868pt,3.375pt);
\draw[color=black] (0.20937720171548772,-2.3) node {$w_1$};
\draw [fill=white,shift={(-0.8660254037844386,-1.5000000000000002)}] (0,0) ++(1.5 pt,1.5pt) -- ++(-3pt,0pt)--++(0pt,-3 pt) -- ++(3pt,0pt) -- ++(0pt,3pt);
\draw[color=black] (-1.1047540490234774,-1.534922403779555) node {$y_1$};
\draw [fill=white,shift={(0.866025403784439,-1.5)}] (0,0) ++(1.5 pt,1.5pt) -- ++(-3pt,0pt)--++(0pt,-3 pt) -- ++(3pt,0pt) -- ++(0pt,3pt);
\draw[color=black] (1.104702635808183,-1.5907406450248789) node {$y_0$};
\end{scriptsize}
\end{tikzpicture}
\caption{\label{intermbis}Vertices $v_0$, $v_1$, $v_2$, $v_3$, $u_0$, $u_1$, $u_3$, $w_0$, $w_1$, $w_3$, $y_0$ and $y_1$.}
\end{center}
\end{figure}

Let $z$ be the neighbor of $v_2$ different from $v_1$, $v_3$ and $y_0$. The only edges that may or not be among $v_0$, $v_1$, $v_2$, $v_3$, $u_0$, $u_1$, $w_1$, $y_0$, $y_1$ and $z$ are $zy_1$ and $zu_1$, and as $G$ is triangle-free, there is at most one of those edges. Let $H^* = G - \{v_0, v_1, v_2, v_3, u_0, u_1, w_1, y_0, y_1, z\}$. Graph $H^*$ has $n - 10$ vertices and $m' \le m - 23$ edges (recall that $u_1$ cannot be adjacent both to $y_0$ and $y_1$, and thus is not adjacent to $y_0$). Adding to any induced forest of $H^*$ the vertices $v_0$, $v_1$, $v_2$, $u_1$ and $w_1$ leads to an induced forest of $G$, so Observation~\ref{abg} applied to $(\alpha,\beta,\gamma) = (10,23,5)$ leads to a contradiction, completing the proof.
\end{proof}

\begin{lemm} \label{33333}
There is no $5$-face with only $3$-vertices in $G$.
\end{lemm}

\begin{proof}
Let $C = v_0v_1v_2v_3v_4$ be such a face, and $u_0$, $u_1$, $u_2$, $u_3$, and $u_4$ be the third neighbors of $v_0$, $v_1$, $v_2$, $v_3$, and $v_4$ respectively. The $u_i$ are all distinct due to the girth assumption and Lemma~\ref{sepcycle}. We will consider the indices of the $u_i$ and $v_i$ modulo $5$. There is no edge $u_iu_{i+1}$ for any $i$ due to Lemma~\ref{2vertices3}. Let $H^* = G + x + y + \{xu_0,xu_1,yu_2,yu_3,xy\} - C$. Graph $H^*$ has $n-3$ vertices and $m-5$ edges. Let $F'$ be any induced forest of $H^*$. Let $F$ be the subgraph of $G$ induced by the vertices of $V(F')\backslash\{x,y\}$, plus the vertices $v_0$ and $v_3$, plus $v_1$ if $x \in V(F')$, and plus $v_3$ if $y \in V(F')$. Subgraph $F$ is an induced forest of $G$. Thus Observation~\ref{abg} applied to $(\alpha,\beta,\gamma) = (3,5,2)$ leads to a contradiction completing the proof.
\end{proof}

\begin{lemm} \label{deg3deg3deg4}
There is no 3-vertex adjacent to a $3$-vertex and to a $4$-vertex in $G$.
\end{lemm}

\begin{proof} 
Let $v$ be a 3-vertex adjacent to a $3$-vertex $u$ and to a $4$-vertex $w$. Let $x$ be the third neighbor of $v$. By Lemma~\ref{deg3deg4}, $x$ and $u$ have a common neighbor distinct from $v$ which contradicts Lemma~\ref{2vertices3}.
\end{proof}

For every face $f$ of $G$, let $l(f)$ be the length of $f$, and let $c_{4^+}(f)$ be the number of $4^+$-vertices in $f$. For every vertex $v$, let $d(v)$ be the degree of $v$. Let $k$ be the number of faces of $G$, and for every $3 \le d \le 5$ and every $4 \le l$, let $k_l$ be the number of faces of length $l$ and $n_d$ the number of $d$-vertices in $G$.

Each $4$-vertex is in the boundary of at most four faces, and each $5$-vertex is in the boundary of at most five faces. Therefore the sum of the $c_{4^+}(f)$ over all the $4$-faces and $5$-faces is $\sum_{f,4 \le l(f) \le 5} c_{4^+}(f) \le 4n_4+5n_5$.
From Lemmas~\ref{deg3deg5},~\ref{33333} and~\ref{deg3deg3deg4} we can deduce that for each $5$-face $f$ we have $c_{4^+}(f) \ge 2$. Moreover, by Lemmas~\ref{2vertices3} and~\ref{3444}, for each $4$-face $f$, $c_{4^+}(f) \ge 4$. Thus $\sum_{f,l(f) = 4} c_{4^+}(f) + \sum_{f,l(f) = 5} c_{4^+}(f) \ge 4k_4 + 2k_5$. Thus we have the following:

$$ 4n_4+5n_5 \ge 4k_4 + 2k_5$$
By Euler's formula, we have:

\begin{eqnarray*} 
-12 & = & 6m - 6n - 6k  \\
& = & 2\sum_{v \in V(G)}d(v) + \sum_{f \in F(G)}l(f) - 6n - 6k  \\
& = & \sum_{d \ge 3}(2d - 6)n_d + \sum_{l \ge 4}(l-6)k_l  \\
& \ge & 2n_4+4n_5 - 2k_4 - k_5  \\
& \ge & 0 \\
\end{eqnarray*}
This is a contradiction, which ends the proof of Theorem~\ref{genmain}.

\section{Proof of Theorem~\ref{bmain}}\label{bproofmain}

The proof of Theorem~\ref{bmain} follows the same scheme as that of Theorem~\ref{main}. 
We will prove the following more general statement than Theorem~\ref{bmain}:

\begin{theo} \label{bgenmain}
If $a$ and $b$ are positive constants such that equations (\ref{ba})--(\ref{b5a9b}) are verified, then $a(G) \ge an-bm$ for all $G \in {\cal P}_5$.
\end{theo}

\begin{eqnarray}
0 \le a \le 1 \label{ba}\\
0 \le b \\
a - 5b \le 0 \label{ba5b}\\
11a - 23b \le 6 \label{b5a9b}
\end{eqnarray}

\begin{figure}[h]
\begin{center}
\definecolor{cqcqcq}{rgb}{0.4,0.4,0.4}
\begin{tikzpicture}[line cap=round,line join=round,>=triangle 45,x=30.0cm,y=30.0cm]
\clip(0.0,0.82) rectangle (0.3,1.04);
\fill[color=cqcqcq,fill=cqcqcq,fill opacity=0.1] (0.0,-0.0) -- (0.1875,0.9375) -- (0.2174,1.0) -- (1.0,1.0) -- (1.0,0.0) -- cycle;
\draw [->] (0.028742864621177887,0.863654393991357) -- (0.028742864621177887,0.9302741455213643);
\draw [->] (0.028742864621177887,0.863654393991357) -- (0.09536261615118535,0.863654393991357);
\draw (-0.2015721282788287,0.5362016543476625) node[anchor=north west] {$a$};
\draw (-0.14323197552215577,0.334298968827021) node[anchor=north west] {$b$};
\draw [domain=0.0:0.3] plot(\x,{(-1.0-0.0*\x)/-1.0});
\draw [domain=0.0:0.3] plot(\x,{(--0.0-1.0*\x)/-0.2});
\draw (0.09142794831552796,0.8613358574204568) node[anchor=north west] {$b$};
\draw (0.014169521995545215,0.9427543528499769) node[anchor=north west] {$a$};
\draw (0.19542967605396627,0.9409714660887465) node[anchor=north west] {$(\frac{3}{16},\frac{15}{16})$};
\draw (0.22,0.995) node[anchor=north west] {$(\frac{5}{23}, 1)$};
\draw (0.0213010690404667,1.028) node[anchor=north west] {\small{$a = 1$}};
\draw (0.18,0.8726274735749158) node[anchor=north west] {\small{$a = 5b$}};
\draw (0.08,0.91) node[anchor=north west] {\small{$11a - 23b = 6$}};
\draw [domain=0.0:0.3] plot(\x,{(-0.016312500000000008-0.0625*\x)/-0.02990000000000001});
\begin{scriptsize}
\draw [fill=black] (0.0,-0.0) circle (1.5pt);
\draw [fill=black] (1.0,0.0) circle (1.5pt);
\draw [fill=black] (0.1875,0.9375) circle (1.5pt);
\draw [fill=black] (1.0,1.0) circle (1.5pt);
\draw [fill=black] (0.2174,1.0) circle (1.5pt);
\end{scriptsize}
\end{tikzpicture}
\end{center}
        \caption{\label{polybis}The top-left part of the polygon of the constraints on $a$ and $b$.}
\end{figure}
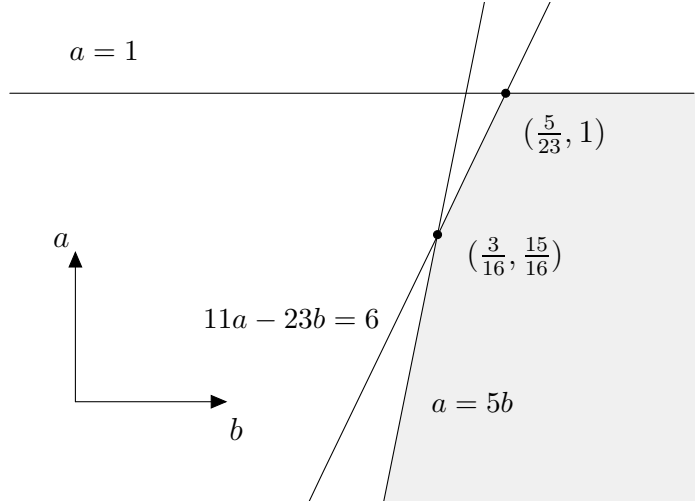

This series of inequalities defines a polygon represented in Figure~\ref{polybis}, and for a graph in ${\cal P}_5$ of given order $n$ and size $m$, the highest lower bound will be given by maximizing $an-bm$ for $a$ and $b$ in this polygon. This maximum will be achieved at a vertex of the polygon. Moreover, by Euler's formula, every planar graph of girth at least $5$, order $n \ge 4$ and size $m$ satisfies $0 \le m \le \frac{5n-10}{3}$. Then for $n \ge 4$ the maximum will always be achieved at the intersection of $11a - 23b = 6$ and $a = 1$. The corresponding intersection is $(b,a) = (\frac{5}{23},1)$, represented in Figure~\ref{polybis}.

Let $G = (V,E)$ be a counter-example to Theorem~\ref{bgenmain} of minimum order. Let $n = |V|$ and $m = |E|$.
We will use the scheme presented in Observation~\ref{babg} for most of our lemmas.

\begin{obs} \label{babg}

Let $\alpha$, $\beta$, $\gamma$ be integers satisfying $\alpha \ge 1$, $\beta \ge 0$, $\gamma \ge 0$ and $a\alpha-b\beta \le \gamma$.

Let $H^* \in {\cal P}_5$ be a graph with $|V(H^*)| = n - \alpha$ and $|E(H^*)| \le m - \beta$.

By minimality of $G$, $H^*$ admits an induced forest of order at least $a(n-\alpha) - b(m - \beta)$.

For all induced forest $F^*$ of $H^*$ of order at least $a(n-\alpha) - b(m - \beta)$, if there is an induced forest $F$ of $G$ of order at least $|V(F^*)| + \gamma$, then we get a contradiction: as $a\alpha - b\beta \le \gamma$, we have $|V(F)| \ge an - bm$.
\end{obs}

Table~\ref{babgtab} contains the values of $(\alpha, \beta, \gamma)$ that will be used throughout this section. For each one, the inequality $a\alpha - b\beta \le \gamma$ is a consequence of the constraints (\ref{ba})--(\ref{b5a9b}).

\begin{table}[h]
\begin{center}
\begin{tabular}{|l|l|l|l|}
        \hline
        $\alpha$ & $\beta$ & $\gamma$ & proof \\
        \hline
        1 & 5 & 0 & (\ref{ba5b}) \\
        2 & 5 & 1 & (\ref{ba}) + (\ref{ba5b}) \\
        3 & 5 & 2 & 2(\ref{ba}) + (\ref{ba5b}) \\
        5 & 10 & 3 & 3(\ref{ba}) + 2(\ref{ba5b}) \\
        1 & 0 & 1 & (\ref{ba}) \\
        6 & 14 & 3 & ((\ref{ba5b}) + (\ref{b5a9b}))/2 \\
        6 & 10 & 4 & 4(\ref{ba}) + 2(\ref{ba5b}) \\
        7 & 14 & 4 & (\ref{ba}) + ((\ref{ba5b}) + (\ref{b5a9b}))/2 \\
        7 & 10 & 5 & 5(\ref{ba}) + 2(\ref{ba5b}) \\
        10 & 15 & 7 & 7(\ref{ba}) + 3(\ref{ba5b})\\
        8 & 14 & 5 & 2(\ref{ba}) + ((\ref{ba5b}) + (\ref{b5a9b}))/2 \\
        10 & 20 & 6 & 6(\ref{ba}) + 4(\ref{ba5b}) \\
        11 & 19 & 7 & 4(\ref{ba}) + (3(\ref{ba5b}) + (\ref{b5a9b}))/2 \\
        12 & 23 & 7 & (\ref{ba}) + (\ref{b5a9b}) \\
        8 & 19 & 4 & 2(\ref{ba}) + (3(\ref{ba5b}) + (\ref{b5a9b}))/2 \\
        9 & 15 & 6 & 6(\ref{ba}) + 3(\ref{ba5b}) \\
        11 & 23 & 6 & (\ref{b5a9b}) \\
        13 & 23 & 8 & 2(\ref{ba}) + (\ref{b5a9b}) \\
        \hline
\end{tabular}
   \caption{\label{babgtab} The various triples ($\alpha$,$\beta$,$\gamma$) and the combinations of inequalities which imply $a\alpha - b\beta \le \gamma$.}
\end{center}
\end{table}
\bigskip

We will now prove a series of lemmas on the structure of $G$.

\begin{lemm} \label{b2co}
Graph $G$ is 2-edge-connected.
\end{lemm}

\begin{proof}
See the proof of Lemma~\ref{2co}.
\end{proof}

\begin{lemm} \label{bdegle4}
Every vertex in $G$ has degree at most $4$.
\end{lemm}

\begin{proof}
By contradiction, suppose $v \in V(G)$ has degree at least $5$. Observation~\ref{babg} applied to $H^* = G - v$, $(\alpha,\beta,\gamma) = (1,5,0)$ and $F = F'$ leads to a contradiction.
\end{proof}

\begin{lemm} \label{bdeg3deg4}
If $v$ is a $3$-vertex adjacent to a $4$-vertex $w$ in $G$, and if $x$ and $y$ are the two other neighbors of $v$, then there are two other vertices $x'$ and $y'$ such that $vxx'y'y$ is a cycle.
\end{lemm}

\begin{proof}
Suppose that there is no cycle as in the statement of the lemma. Let $H^* = G + xy - \{w,v\}$. Graph $H^*$ has $n-2$ vertices and $m' \le m - 5$ edges. As there are no $x'$ and $y'$ as in the lemma, adding the edge $xy$ does not create any $4^-$-cycle in $H^*$, and thus $H^*\in {\cal P}_5$. Let $F'$ be any induced forest of $H^*$. Adding $v$ to $F'$ leads to a forest of $G$. Observation~\ref{babg} applied to $(\alpha,\beta,\gamma) = (2,5,1)$ completes the proof.
\end{proof}

\begin{lemm} \label{bdeg2deg4}
There is no $2$-vertex adjacent to a $4$-vertex in $G$.
\end{lemm}

\begin{proof}
Let $v$ be a $2$-vertex and $w$ a $4$-vertex adjacent to $v$. Let $H^* = G - \{v,w\}$. Graph $H^*$ has $n-2$ vertices and $m' = m-5$ edges. Let $F'$ be any induced forest of $H^*$. Adding $v$ to $F'$ leads to an induced forest of $G$. Observation~\ref{babg} applied to $(\alpha,\beta,\gamma) = (2,5,1)$ completes the proof.
\end{proof}

\begin{lemm} \label{bdeg2deg3deg2}
There is no $3$-vertex adjacent to two $2$-vertices in $G$.
\end{lemm}

\begin{proof}
Let $v$ be a $3$-vertex adjacent to two $2$-vertices $u$ and $w$. Let $H^* = G - \{u,v,w\}$. Graph $H^*$ has $n-3$ vertices and $m' = m-5$ edges. Let $F'$ be any induced forest of $H^*$. Adding $u$ and $w$ to $F'$ leads to an induced forest of $G$. Observation~\ref{babg} applied to $(\alpha,\beta,\gamma) = (3,5,2)$ completes the proof.
\end{proof}

\begin{lemm} \label{bsepcyclefirst}
There is no separating $5$-cycles with only $3$-vertices in $G$.
\end{lemm}

\begin{proof}
Let $C = v_0v_1v_2v_3v_4$ be such a cycle. W.l.o.g. $v_0$ has his third neighbor in the interior of $C$ and $v_1$ in the exterior of it. Let $H^* = G - C$. Graph $H^*$ has $n-5$ vertices and $m' = m-10$ edges. Adding $v_0$, $v_1$ and $v_3$ to any induced forest of $H^*$ leads to an induced forest of $G$. Observation~\ref{babg} applied to $(\alpha,\beta,\gamma) = (5,10,3)$ leads to a contradiction.
\end{proof}

\begin{lemm} \label{bdegge3}
Every vertex in $G$ has degree at least $3$.
\end{lemm}

\begin{proof}
Let $v$ be a $2$-vertex in $G$. 

Suppose that $v$ is adjacent to a $2$-vertex $u$ and a $3$-vertex $w$. Let $H^* = G - \{u,v,w\}$. Graph $H^*$ has $n-3$ vertices and $m' = m - 5$ edges. Let $F'$ be any induced forest of $H^*$. Adding $u$ and $v$ to $F'$ leads to an induced forest of $G$. Observation~\ref{babg} applied to $(\alpha,\beta,\gamma) = (3,5,2)$ leads to a contradiction.

Suppose that $v$ is adjacent to two $3$-vertices $u$ and $w$. Consider two cases according to the presence or not of $5$-cycles containing $uvw$.

\begin{itemize}
        \item Suppose there is no $5$-cycle containing $uvw$. Let $H^* = G + uw - v$. Graph $H^*$ has $n-1$ vertices and $m-1$ edges. As there is no $5$-cycle containing $uvw$, adding the edge $uw$ does not create any cycle of length $3$ or $4$ in $H^*$, thus $H^* \in {\cal P}_5$. Let $F'$ be any induced forest of $H^*$. Adding $v$ to $F'$ leads to an induced forest of $G$. Observation~\ref{babg} applied to $(\alpha,\beta,\gamma) = (1,0,1)$ leads to a contradiction.

        \item Suppose there is a $5$-cycle containing $uvw$, say $uvwxy$. By Lemma~\ref{bdeg2deg3deg2}, both $x$ and $y$ are $3^+$-vertices. 

Suppose $x$ or $y$, say $x$, has degree $3$, and the other one has degree $4$. Let $H^* = G - \{u,v,w,x,y\}$. Graph $H^*$ has $n-5$ vertices and, since there is no chord in the $5$-cycle, $m' = m-10$ edges. Let $F'$ be any induced forest of $H^*$. Adding $u$, $v$ and $x$ to $F'$ leads to an induced forest of $G$. Observation~\ref{babg} applied to $(\alpha,\beta,\gamma) = (5,10,3)$ leads to a contradiction.

Suppose both $x$ and $y$ have degree $3$. Let $u'$, $w'$, $x'$, and $y'$ be the third neighbors of $u$, $w$, $x$ and $y$ respectively. They are all distinct by the girth assumption. By Lemma~\ref{bdeg2deg3deg2}, $u'$ and $w'$ are $3^+$-vertices. Suppose $x'$ or $y'$, say $x'$, has degree $2$. Let $H^* = G - \{u,v,w,x,y,x'\}$. Graph $H'$ has $n-6$ vertices and $m' = m-10$ edges. Adding $u$, $v$, $x$ and $x'$ to any induced forest of $H^*$ leads to an induced forest of $G$. Observation~\ref{babg} applied to $(\alpha,\beta,\gamma) = (6,10,4)$ leads to a contradiction.

Hence $u'$, $w'$, $x'$ and $y'$ are $3^+$-vertices. Suppose $u'$ or $y'$ is a $4$-vertex. By the girth assumption, $u'y' \notin E$. Let $H^* = G - \{u,v,w,x,y,u',y'\}$. Graph $H'$ has $n-7$ vertices and $m' \le m - 14$ edges. Adding $u$, $v$, $w$ and $y$ to any induced forest of $H^*$ leads to an induced forest of $G$. Observation~\ref{babg} applied to $(\alpha,\beta,\gamma) = (7,14,4)$ leads to a contradiction. Therefore $u'$, $w'$, $x'$ and $y'$ are $3$-vertices.

Let us now show that $u'x' \notin E$ (and by symmetry $w'y' \notin E$). Suppose by contradiction that $u'x' \in E$. By Lemma~\ref{bsepcyclefirst}, the cycle $uyxx'u'$ bounds a face, hence the cycle $uvwxx'u'$ separates $y'$ from the third neighbor of $x'$. Let $H^* = H - \{u,v,w,x,y,u',x'\}$. Graph $H^*$ has $n-7$ vertices and $m' \le m - 10$ edges. Adding $u$, $v$, $x$, $y$ and $x'$ to any induced forest of $H^*$ leads to an induced forest of $G$. Observation~\ref{babg} applied to $(\alpha,\beta,\gamma) = (7,10,5)$ leads to a contradiction.

Suppose that there is no vertex adjacent to both $u'$ and $y'$. Let $H^* = G - \{u,v,w\} + u'y$. Graph $H^*$ has $n-3$ vertices and $n - 5$ edges, and has girth at least $5$ since $u'x' \notin E$ and there is no vertex adjacent to $u'$ and $y'$. Adding $u$ and $v$ to any induced forest of $H^*$ leads to an induced forest of $G$. Observation~\ref{babg} applied to $(\alpha,\beta,\gamma) = (3,5,2)$ leads to a contradiction. Hence there is a vertex $z$ adjacent to $u'$ and $y'$.

Suppose that there is no vertex adjacent to $x'$ and $y'$. Let $H^* = G - \{v,w,x\} + x'y$. Graph $H^*$ has $n-3$ vertices and $n - 5$ edges, and has girth at least $5$ since $u'x' \notin E$ and there is no vertex adjacent to $x'$ and $y'$. Adding $x$ and $v$ to any induced forest of $H^*$ leads to an induced forest of $G$. Observation~\ref{babg} applied to $(\alpha,\beta,\gamma) = (3,5,2)$ leads to a contradiction.
Hence there is a vertex $z'$ adjacent to $x'$ and $y'$.

 Suppose $z$ is a $2$-vertex. Vertices $z$ and $z'$ are distinct, and non-adjacent. Let $H^* = G - \{u,v,w,x,y,u',x',y',z,z'\}$. Graph $H^*$ has $n-10$ vertices and $n - 15$ edges. Adding $u$, $v$, $x$, $u'$, $x'$, $y'$ and $z$ to any induced forest of $H^*$ leads to an induced forest of $G$. Observation~\ref{babg} applied to $(\alpha,\beta,\gamma) = (10,15,7)$ leads to a contradiction.

Therefore $z$ is a $3^+$-vertex. Let $H^* = G - \{u,v,w,x,y,u',y',z\}$. Graph $H^*$ has $n-8$ vertices and $n - 14$ edges ($w' \ne z$, since $w'y' \notin E$). Adding $u$, $v$, $x$, $u'$, and $y'$ to any induced forest of $H^*$ leads to an induced forest of $G$. Observation~\ref{babg} applied to $(\alpha,\beta,\gamma) = (8,14,5)$ leads to a contradiction.

Therefore $x$ and $y$ have degree $4$. By Lemma~\ref{bdeg3deg4}, there is an other $5$-cycle containing $uvw$, and as $G$ has girth at least $5$, there are $x\rq{}$ and $y\rq{}$ distinct from all the vertices defined previously such that $uvwx\rq{}y\rq{}$ is a cycle. By symmetry, $x\rq{}$ and $y\rq{}$ are $4$-vertices. Let $H^* = G - \{u,v,w,x,y,x\rq{}\}$. Graph $H^*$ has $n-6$ vertices and $m' \le m-14$ edges. Let $F'$ be any induced forest of $H^*$. Adding $u$, $v$ and $w$ to $F'$ leads to an induced forest of $G$. Observation~\ref{babg} applied to $(\alpha,\beta,\gamma) = (6,14,3)$ leads to a contradiction.

\end{itemize}

Therefore by Lemmas~\ref{b2co},~\ref{bdegle4}, and~\ref{bdeg2deg4}, every $2$-vertex is only adjacent to $2$-vertices, so either $G$ does not have any $2$-vertex, or it is $2$-regular. If $G$ is $2$-regular, then $G$ is a $n$-cycle and thus $m = n$. Since $G \in {\cal P}_5$, we have $n \ge 5$. It is clear that $G$ has an induced forest of size $n-1$. Recall that $a \le 5b$ and $a \le 1$; this gives that $5(a-b) \le 4$. Since $n \ge 5$, we can deduce that $an-bm = (a-b)n \le n-1$. This contradicts the fact that $G$ is a counter-example. Therefore, $G$ has minimum degree at least $3$. This completes the proof.
\end{proof}

\begin{lemm} \label{b4link4}
Let $v_0v_1v_2v_3v_4$ be a $5$-cycle in $G$ such that $v_0$ is a $4$-vertex and the other $v_i$ are $3$-vertices. The third neighbors of $v_1$ and $v_2$ are $3$-vertices.
\end{lemm}

\begin{proof}
Let $v_0v_1v_2v_3v_4$ be a $5$-cycle in $G$ such that $v_0$ is a $4$-vertex and the other $v_i$ are $3$-vertices. Let $u_i$ be the third neighbor of $v_i$ for $i \in \{1,2,3,4\}$. Suppose $u_1$ or $u_2$, say $u_{i_0}$, is a $4$-vertex. 
Let $H^* = G - C - u_{i_0}$. Graph $H^*$ has $n-6$ vertices and $m\rq{} = m-14$ edges. Adding $v_1$, $v_2$ and $v_4$ to any induced forest of $H^*$ leads to an induced forest of $G$. Observation~\ref{babg} applied to $(\alpha,\beta,\gamma) = (6,14,3)$ completes the proof.
\end{proof}

\begin{lemm} \label{bsepcycle}
There is no separating $5$-cycles with at most one $4$-vertex in $G$.
\end{lemm}

\begin{proof}
Let $C = v_0v_1v_2v_3v_4$ be such a cycle. By Lemma~\ref{bsepcyclefirst}, $C$ has exactly one $4$-vertex, say $v_0$. Let $u_i$ be the third neighbor of $v_i$ for $i \in \{1,2,3,4\}$. By the girth assumption, all the $u_i$ are distinct. By Lemma~\ref{b4link4}, all the $u_i$ have degree $3$.

Suppose $C$ separates $u_1$ and $u_2$. Let $H^* = G - C$. Graph $H^*$ has $n-5$ vertices and $m' \le m-10$ edges, and adding $v_1$, $v_2$ and $v_4$ to any induced forest of $H^*$ leads to an induced forest of $G$. Observation~\ref{babg} applied to $(\alpha,\beta,\gamma) = (5,10,3)$ leads to a contradiction.

So $C$ does not separate $u_1$ and $u_2$, and by symmetry it does not separate $u_3$ and $u_4$ either.

Suppose $C$ separates some of the $u_i$. Say $u_1$ and $u_2$ are in the interior of $C$ w.l.o.g., and $u_3$ and $u_4$ are in the exterior of $C$. By Lemma~\ref{bdeg3deg4} there is a vertex $w$ such that $u_1v_1v_2u_2w$ is a cycle. Since $u_1$, $v_1$, $v_2$ and $u_2$ have degree $3$, and $v_0$ has degree $4$, $w$ has degree $3$ by Lemma~\ref{b4link4}. Vertex $w$ cannot be adjacent to $v_0$, $v_1$ or $v_2$ by the girth assumption, and it cannot be adjacent to $v_3$, $v_4$, $u_3$ or $u_4$ by planarity. Let $w'$ be the third neighbor of $u_1$. It is also non-adjacent to all the vertices defined previously (except for $u_1$) by the girth assumption and planarity of $G$. Let $H^* = G - C - \{u_1,u_2,u_3,w,w'\}$. Graph $H^*$ has $n-10$ vertices and $m' \le m-20$ edges, and adding $v_1$, $v_2$, $v_3$, $v_4$, $u_1$ and $u_2$ to any induced forest of $H^*$ leads to an induced forest of $G$. Observation~\ref{babg} applied to $(\alpha,\beta,\gamma) = (10,20,6)$ leads to a contradiction.

Therefore $C$ does not separate any of the $u_i$, say the $u_i$ are in
the exterior of $C$ up to changing the plane embedding. Then as $G$ is
$2$-edge-connected by Lemma~\ref{b2co}, the two neighbors of $v_0$
distinct from $v_1$ and $v_4$ are in the interior of $C$. By
Lemma~\ref{bdeg3deg4}, either $u_1u_3 \in E$, or there is a vertex $w$
such that $u_1v_1v_2u_2w$ is a cycle. If $u_1u_3 \in E$, then the cycle $v_1v_2v_3u_3u_1$ is separating with only $3$-vertices, contradicting Lemma~\ref{bsepcyclefirst}. Thus $u_1u_3 \notin E$ (and $u_2u_4 \notin E$ by symmetry), and there is a vertex $w$ such that $u_1v_1v_2u_2w$ is a cycle. Since $u_1$, $v_1$, $v_2$ and $u_2$ have degree $3$, and $v_0$ has degree $4$, by Lemma~\ref{b4link4} $w$ has degree $3$. If $w = u_4$, then $u_2u_4 \in E$, which is impossible; hence $w$ is not adjacent to $v_4$. It is not adjacent to the other $v_i$ by girth assumption. Let $H^* = G - \{v_1,v_2,v_3,v_4,u_1,u_2,w\}$. Graph $H^*$ has $n-7$ vertices and $m' \le m-14$ edges. Let $F'$ be any induced forest of $H^*$. Adding $u_1$, $u_2$, $v_2$, and $v_4$ to $F'$ leads to an induced forest of $G$. Observation~\ref{babg} applied to $(\alpha,\beta,\gamma) = (7,14,4)$ completes the proof.
\end{proof}

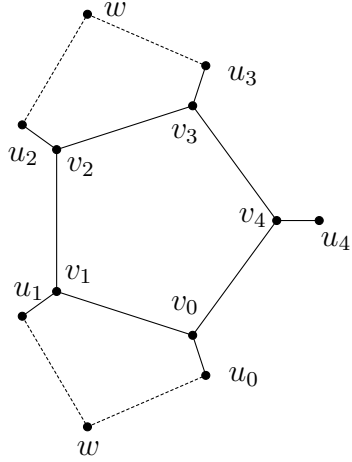
\begin{figure}[h]
\begin{center}
\begin{tikzpicture}[line cap=round,line join=round,>=triangle 45,x=1.6cm,y=1.6cm]
\clip(-1.5,-2.3) rectangle (1.7,2.3);
\draw (1.0,-0.0)-- (0.3090169943749475,0.9510565162951535);
\draw (0.3090169943749475,0.9510565162951535)-- (-0.8090169943749473,0.5877852522924731);
\draw (-0.8090169943749473,0.5877852522924731)-- (-0.8090169943749475,-0.5877852522924731);
\draw (0.30901699437494745,-0.9510565162951536)-- (1.0,-0.0);
\draw (0.30901699437494745,-0.9510565162951536)-- (-0.8090169943749475,-0.5877852522924731);
\draw (1.3492549017242996,-0.0)-- (1.0,-0.0);
\draw (0.3090169943749475,0.9510565162951535)-- (0.41694269437650827,1.2832176664280721);
\draw (-1.091570145238658,0.7930721328168736)-- (-0.8090169943749473,0.5877852522924731);
\draw (-1.091570145238658,-0.7930721328168735)-- (-0.8090169943749475,-0.5877852522924731);
\draw (0.4169426943765081,-1.2832176664280721)-- (0.30901699437494745,-0.9510565162951536);
\draw [dash pattern=on 1pt off 1pt] (-1.091570145238658,-0.7930721328168735)-- (-0.5555264139055502,-1.709734498543117);
\draw [dash pattern=on 1pt off 1pt] (-0.5555264139055502,-1.709734498543117)-- (0.4169426943765081,-1.2832176664280721);
\draw [dash pattern=on 1pt off 1pt] (-1.091570145238658,0.7930721328168736)-- (-0.55552641390555,1.709734498543117);
\draw [dash pattern=on 1pt off 1pt] (-0.55552641390555,1.709734498543117)-- (0.41694269437650827,1.2832176664280721);
\draw (0.041909843601526,-0.5337510051834266) node[anchor=north west] {$v_0$};
\draw (-0.8328104613888772,-0.27751980473169224) node[anchor=north west] {$v_1$};
\draw (-0.8151393441163438,0.6148716175312444) node[anchor=north west] {$v_2$};
\draw (0.041909843601526,0.8887739352555121) node[anchor=north west] {$v_3$};
\draw (0.5985500376863281,0.19076480299044285) node[anchor=north west] {$v_4$};
\draw (-1.3,0.7) node[anchor=north west] {$u_2$};
\draw (-1.2541084209286124,-0.41888874291195943) node[anchor=north west] {$u_1$};
\draw (0.5130033063247273,-1.1257334338132954) node[anchor=north west] {$u_0$};
\draw (1.2728613490436635,-0.003617487007424538) node[anchor=north west] {$u_4$};
\draw (0.5041677476884607,1.365894101613914) node[anchor=north west] {$u_3$};
\draw (-0.5119214954822098,1.9) node[anchor=north west] {$w$};
\draw (-0.7328104613888772,-1.7353869797156976) node[anchor=north west] {$w$};
\begin{scriptsize}
\draw [fill=black] (0.3090169943749475,0.9510565162951535) circle (1.5pt);
\draw [fill=black] (1.0,-0.0) circle (1.5pt);
\draw [fill=black] (-0.8090169943749473,0.5877852522924731) circle (1.5pt);
\draw [fill=black] (-0.8090169943749475,-0.5877852522924731) circle (1.5pt);
\draw [fill=black] (0.30901699437494745,-0.9510565162951536) circle (1.5pt);
\draw [fill=black] (-1.091570145238658,0.7930721328168736) circle (1.5pt);
\draw [fill=black] (0.41694269437650827,1.2832176664280721) circle (1.5pt);
\draw [fill=black] (1.3492549017242996,-0.0) circle (1.5pt);
\draw [fill=black] (0.4169426943765081,-1.2832176664280721) circle (1.5pt);
\draw [fill=black] (-1.091570145238658,-0.7930721328168735) circle (1.5pt);
\draw [fill=black] (-0.55552641390555,1.709734498543117) circle (1.5pt);
\draw [fill=black] (-0.5555264139055502,-1.709734498543117) circle (1.5pt);
\end{scriptsize}
\end{tikzpicture}
\end{center}
\caption{The construction of Lemma~\ref{b3facesadj}. At least one of the two $w$ represented exists.\label{bfigureplus}}
\end{figure}

\begin{lemm} \label{b3facesadj}
Let $C = v_0v_1v_2v_3v_4$ be a $5$-cycle in $G$ with only $3$-vertices, and $u_i$ be the third neighbor of $v_i$ for $i \in \{0,1,2,3,4\}$. Then there is a vertex $w$ adjacent either to $u_0$ and $u_1$ or to $u_2$ and $u_3$.
\end{lemm}

\begin{proof}
Let $C = v_0v_1v_2v_3v_4$ be a $5$-cycle with only vertices of degree $3$ in $G$, and let $u_i$ be the third neighbor of $v_i$ for $i \in \{0,1,2,3,4\}$. See Figure~\ref{bfigureplus} for an illustration of the statement of the lemma. By Lemma~\ref{bsepcyclefirst}, $C$ is the boundary of a face.

Let us first show that no two $u_i$ can be adjacent. Suppose two of the $u_i$ are adjacent. By the girth assumption, w.l.o.g. $u_0u_2 \in E$. Then by Lemma~\ref{bsepcycle}, $u_0$ and $u_2$ have degree $4$. Let $H^* = G - C - \{u_0, u_2\}$. Graph $H^*$ has $n-7$ vertices and $m\rq{} \le m-14$ edges. Let $F'$ be any induced forest of $H^*$. Adding $v_0$, $v_1$, $v_2$ and $v_3$ to $F\rq{}$ leads to an induced forest of $G$ by planarity. Observation~\ref{babg} applied to $(\alpha,\beta,\gamma) = (7,14,4)$ leads to a contradiction.

Suppose by contradiction that there is no vertex $w$ adjacent either to $u_0$ and $u_1$, or to $u_2$ and $u_3$. Let $H^* = G - C + \{x,y\} + \{u_0x,u_1x,u_2y,u_3y,xy\}$. Graph $H^*$ is of girth at least $5$ by hypothesis and because the $u_i$ are not adjacent. Graph $H^*$ has $n-3$ vertices and $m\rq{} = m - 5$ edges. Let $F'$ be any induced forest of $H^*$. Removing $x$ and $y$, adding $v_0$ and $v_3$, plus $v_1$ if $x \in F'$, and $v_2$ if $y \in F'$ to $F\rq{}$ leads to an induced forest of $G$. Observation~\ref{babg} applied to $(\alpha,\beta,\gamma) = (3,5,2)$ completes the proof.
\end{proof}

\begin{lemm} \label{b43333}
There is no $5$-face with exactly one $4$-vertex in $G$.
\end{lemm}

\begin{proof}
Let $C = v_0v_1v_2v_3v_4$ be such a face, with $v_0$ the $4$-vertex, and let $u_i$ be the third neighbor of $v_i$ for $i \in \{1,2,3,4\}$. By Lemma~\ref{b4link4}, the $u_i$ have degree $3$. The $u_i$ are all distinct and not adjacent to $v_0$ by the girth assumption. By Lemma~\ref{bdeg3deg4}, either $u_1u_3 \in E$, or there is a vertex adjacent to both $u_1$ and $u_2$. However in the former case, the cycle $u_1v_1v_2v_3u_3$ is a separating cycle with five vertices of degree $3$, contradicting Lemma~\ref{bsepcyclefirst}. Hence $u_1u_3 \notin E$ and $u_2u_4 \notin E$ by symmetry. We also have $u_1u_4 \notin E$ by Lemma~\ref{bsepcycle} applied to $u_1v_1v_0v_4u_4$. Let $w$ be the vertex adjacent to both $u_1$ and $u_2$. By Lemma~\ref{b4link4}, $w$ has degree $3$. By the girth assumption, $wv_0 \notin E$ and $wv_3 \notin E$. By Lemma~\ref{bsepcyclefirst}, $v_1v_2u_2wu_1$ is the boundary of a face. Moreover, $wv_4 \notin E$ and $wu_3 \notin E$ by applying Lemma~\ref{bsepcycle} to the cycle $wv_4v_3v_2u_2$ and $wu_3v_3v_2u_2$ respectively. By symmetry, let $w\rq{}$ ($\ne w$) be the vertex adjacent to $u_3$ and $u_4$. Vertex $w$ has degree $3$, $w\rq{}v_0 \notin E$, $w\rq{}v_1 \notin E$, $w\rq{}v_2 \notin E$, $w\rq{}u_2 \notin E$ and $u_4v_4v_3u_3w\rq{}$ is the boundary of a face.

Observe now that $wu_4 \notin E$ and $w\rq{}u_1 \notin E$ (by symmetry). By contradiction assume $wu_4 \in E$. Consider $H^* = G - \{v_0, v_1, v_2, v_4, u_1, u_4, w\}$ which has $n-7$ vertices and $m' \le m-14$ edges. Adding the vertices $w$, $u_1$, $v_1$ and $v_4$ to any induced forest of $H^*$ leads to an induced forest of $G$. Observation~\ref{babg} applied to $(\alpha,\beta,\gamma) = (7,14,4)$ completes the proof.

Observe now that $ww' \notin E$. Otherwise, consider $H^* = G - \{v_0, v_1, v_4, u_1,\\ u_4, w, w'\}$ which has $n-7$ vertices and $m' \le m-14$ edges. Adding the vertices $u_1$, $v_1$, $v_4$ and $w'$ to any induced forest of $H^*$ leads to an induced forest of $G$. Observation~\ref{babg} applied to $(\alpha,\beta,\gamma) = (7,14,4)$ completes the proof.

See Figure~\ref{bfig43333} for a summary of the edges between the vertices $v_0$, $v_1$, $v_2$, $v_3$, $v_4$, $u_1$, $u_2$, $u_3$, $u_4$, $w$ and $w'$.

\begin{figure}[h]
\begin{center}
\begin{tikzpicture}[line cap=round,line join=round,>=triangle 45,x=1.5cm,y=1.5cm]
\clip(-2.2,-2.6) rectangle (2.4,2.6);
\draw (1.0,-0.0)-- (0.3090169943749475,0.9510565162951535);
\draw (0.3090169943749475,0.9510565162951535)-- (-0.8090169943749473,0.5877852522924731);
\draw (-0.8090169943749473,0.5877852522924731)-- (-0.8090169943749475,-0.5877852522924731);
\draw (0.30901699437494745,-0.9510565162951536)-- (1.0,-0.0);
\draw (0.30901699437494745,-0.9510565162951536)-- (-0.8090169943749475,-0.5877852522924731);
\draw (0.3090169943749475,0.9510565162951535)-- (0.633532429705718,1.9498123291719263);
\draw (-1.658609433944873,1.2050502911118488)-- (-0.8090169943749473,0.5877852522924731);
\draw (-1.658609433944873,-1.2050502911118486)-- (-0.8090169943749475,-0.5877852522924731);
\draw (0.6335324297057179,-1.9498123291719265)-- (0.30901699437494745,-0.9510565162951536);
\draw (-1.658609433944873,-1.2050502911118486)-- (-0.8142568281748588,-2.506024835106301);
\draw (-0.8142568281748588,-2.506024835106301)-- (0.6335324297057179,-1.9498123291719265);
\draw (-1.658609433944873,1.2050502911118488)-- (-0.8142568281748587,2.5060248351063015);
\draw (-0.8142568281748587,2.5060248351063015)-- (0.633532429705718,1.9498123291719263);
\draw (0.5543722445049945,0.20843592026297625) node[anchor=north west] {$v_0$};
\draw (0.0242387263289926,0.9771295216181791) node[anchor=north west] {$v_1$};
\draw (-0.8416460200251439,0.685556086621378) node[anchor=north west] {$v_2$};
\draw (-0.8416460200251439,-0.27751980473169224) node[anchor=north west] {$v_3$};
\draw (0.0242387263289926,-0.5690932397284933) node[anchor=north west] {$v_4$};
\draw (-1.8842419391046145,1.2068540461611132) node[anchor=north west] {$u_2$};
\draw (0.6338922722313949,2.063903233878983) node[anchor=north west] {$u_1$};
\draw (0.5897144790500613,-1.7353869797156976) node[anchor=north west] {$u_4$};
\draw (-2.016775318648615,-0.8783377919978278) node[anchor=north west] {$u_3$};
\draw (-0.930001606387811,2.4880100484197847) node[anchor=north west] {$w$};
\draw (-0.9830149582054111,-2.0888093251663657) node[anchor=north west] {$w'$};
\draw (1.0,-0.0)-- (2.0128031684147327,0.389557266264276);
\draw (1.0,-0.0)-- (2.0128031684147327,-0.389557266264276);
\begin{scriptsize}
\draw [fill=black] (0.3090169943749475,0.9510565162951535) circle (1.5pt);
\draw [fill=black] (1.0,-0.0) circle (1.5pt);
\draw [fill=black] (-0.8090169943749473,0.5877852522924731) circle (1.5pt);
\draw [fill=black] (-0.8090169943749475,-0.5877852522924731) circle (1.5pt);
\draw [fill=black] (0.30901699437494745,-0.9510565162951536) circle (1.5pt);
\draw [fill=black] (-1.658609433944873,1.2050502911118488) circle (1.5pt);
\draw [fill=black] (0.633532429705718,1.9498123291719263) circle (1.5pt);
\draw [fill=black] (0.6335324297057179,-1.9498123291719265) circle (1.5pt);
\draw [fill=black] (-1.658609433944873,-1.2050502911118486) circle (1.5pt);
\draw [fill=black] (-0.8142568281748587,2.5060248351063015) circle (1.5pt);
\draw [fill=black] (-0.8142568281748588,-2.506024835106301) circle (1.5pt);
\draw [fill=black] (2.0128031684147327,0.389557266264276) circle (1.5pt);
\draw [fill=black] (2.0128031684147327,-0.389557266264276) circle (1.5pt);
\end{scriptsize}
\end{tikzpicture}
\end{center}
\caption{The vertices $v_0$, $v_1$, $v_2$, $v_3$, $v_4$, $u_1$, $u_2$, $u_3$, $u_4$, $w$ and $w'$, and the edges between these vertices. All the vertices except for $v_0$ are $3$-vertices.\label{bfig43333}}
\end{figure}
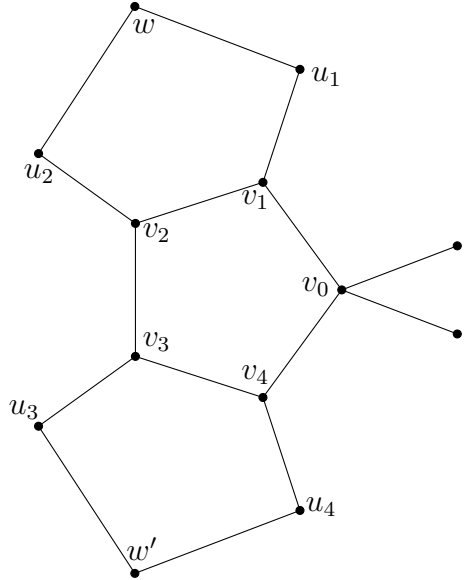

Let $x$ be the third neighbor of $u_1$ ($x$ is distinct from all previously defined vertices). By the girth assumption $xw \notin E$, $xu_2 \notin E$ and $xv_0 \notin E$.

Observe that $xu_4 \notin E$ and $xw' \notin E$. Otherwise consider $H^* = G - \{v_1, v_2, v_3, v_4, u_1, u_2, u_3, u_4, w, w', x\}$, which has $n-11$ vertices and $m' \le m-19$ edges. Adding the vertices $v_1$, $v_2$, $v_3$, $u_1$, $u_4$, $w$ and $w'$ to any induced forest of $H^*$ leads to an induced forest of $G$. Observation~\ref{babg} applied to $(\alpha,\beta,\gamma) = (11,19,7)$ completes the proof. 

Similarly, $xu_3 \notin E$ (just add $u_3$ to $F'$ instead of $w'$).

Finally, let $H^* = G - C - \{u_1, u_2, u_3, u_4, w, w', x\}$. Graph $H^*$ has $n-12$ vertices and $m' \le m-23$ edges. Adding $v_1$, $v_2$, $v_3$, $v_4$, $u_1$, $w$ and $w'$ to any induced forest of $H^*$ leads to an induced forest of $G$. Observation~\ref{babg} applied to $(\alpha,\beta,\gamma) = (12,23,7)$ completes the proof.
\end{proof}

\begin{lemm} \label{b33333linked444}
There is no $5$-face $v_0v_1v_2v_3v_4$ in $G$ such that all the $v_i$ are $3$-vertices, and three of the $v_i$ have a $4$-vertex as their third neighbor.
\end{lemm}

\begin{proof}
Let $C = v_0v_1v_2v_3v_4$ be such a face, and let $u_i$ be the third neighbor of $v_i$ for $i \in \{0,1,2,3,4\}$. 

Suppose two of the $u_i$ are adjacent. By the girth assumption the corresponding $v_i$ are not adjacent. W.l.o.g., say $u_0$ and $u_2$ are adjacent. Then since $C$ is a face, $v_0v_1v_2u_2u_0$ is separating, and thus by Lemma~\ref{bsepcycle}, $u_0$ and $u_2$ have degree $4$. Let $H^* = G - \{v_0,v_1,v_2,v_3,u_0,u_2\}$. Graph $H^*$ has $n-6$ vertices and $m' \le m-14$ edges. Adding $v_0$, $v_1$ and $v_2$ to any induced forest of $H^*$ leads to an induced forest of $G$. Observation~\ref{babg} applied to $(\alpha,\beta,\gamma) = (6,14,3)$ leads to a contradiction. Therefore no two $u_i$ are adjacent.

Let $H^*$ obtained from $G$ where we remove $C$ and three $u_i$ of degree $4$. Graph $H^*$ has $n-8$ vertices and $m' \le m-19$ edges. Let $F'$ be any induced forest of $H^*$. Adding the three $v_i$ that correspond to the $u_i$ we removed, plus another $v_i$ to $F'$ leads to an induced forest of $G$. Observation~\ref{babg} applied to $(\alpha,\beta,\gamma) = (8,19,4)$ completes the proof.
\end{proof}

\begin{lemm}\label{bnoadj}
If there are two $5$-cycles $C = v_0v_1v_2v_3v_4$ and $C' = v_0v_1u_2u_3u_4$ sharing an edge $v_0v_1$ in $G$ with only $3$-vertices, then for all $x \in \{u_2, u_3, u_4\}$, $xv_3 \notin E$. Moreover, for all $x \in \{u_2, u_3, u_4\}$, $x$ and $v_3$ do not share a common neighbor.
\end{lemm}

\begin{proof}
Let $C = v_0v_1v_2v_3v_4$, $C' = v_0v_1u_2u_3u_4$, and $x \in \{u_2, u_3, u_4\}$. Cycles $C$ and $C\rq{}$ are the boundaries of faces by Lemma~\ref{bsepcycle}. If $x$ is either $u_2$ or $u_4$, then we can conclude by the girth assumption and Lemma~\ref{bsepcycle}.

Consider now the case $x = u_3$. By Lemma~\ref{bsepcyclefirst}, $v_3u_3\notin E$. Finally assume that there is a vertex $w$ adjacent to both $v_3$ and $u_3$. Let $H^* = G-(C\cup C') - w$. Graph $H^*$ has $n-9$ vertices and $m' \le m-15$ edges. Adding $v_0$, $v_1$, $v_2$, $v_3$, $u_3$ and $u_4$ to any induced forest of $H^*$ leads to an induced forest of $G$. Observation~\ref{babg} applied to $(\alpha,\beta,\gamma) = (9,15,6)$ completes the proof.
\end{proof}

\begin{figure}[h]
\begin{center}
\begin{tikzpicture}[line cap=round,line join=round,>=triangle 45,x=2.0cm,y=2.0cm]
\clip(-3.3,-1.7) rectangle (1.7,1.7);
\draw (1.0,-0.0)-- (0.3090169943749475,0.9510565162951535);
\draw (0.3090169943749475,0.9510565162951535)-- (-0.8090169943749473,0.5877852522924731);
\draw (-0.8090169943749473,0.5877852522924731)-- (-0.8090169943749475,-0.5877852522924731);
\draw (0.30901699437494745,-0.9510565162951536)-- (1.0,-0.0);
\draw (0.30901699437494745,-0.9510565162951536)-- (-0.8090169943749475,-0.5877852522924731);
\draw (1.3492549017242996,-0.0)-- (1.0,-0.0);
\draw (-0.7769882545702107,-0.35984868745915884) node[anchor=north west] {$v_0$};
\draw (-0.7416460200251439,0.5678849693488446) node[anchor=north west] {$v_1$};
\draw (0.0595809608740594,0.8859650802544457) node[anchor=north west] {$v_2$};
\draw (0.6073855963225948,0.12610703753550965) node[anchor=north west] {$v_3$};
\draw (0.041909843601526,-0.6514221224559599) node[anchor=north west] {$v_4$};
\draw (-2.618033988749895,0)-- (-1.9270509831248428,0.9510565162951538);
\draw (-1.9270509831248428,-0.9510565162951534)-- (-2.618033988749895,0);
\draw (-2.9672888904741948,0)-- (-2.618033988749895,0);
\draw (-0.8090169943749473,0.5877852522924731)-- (-1.9270509831248428,0.9510565162951538);
\draw (-0.8090169943749475,-0.5877852522924731)-- (-1.9270509831248428,-0.9510565162951534);
\draw (0.3090169943749475,0.9510565162951535)-- (0.41694269437650827,1.2832176664280721);
\draw (0.30901699437494745,-0.9510565162951536)-- (0.4169426943765081,-1.2832176664280721);
\draw (-2.034976683126403,1.2832176664280723)-- (-1.9270509831248428,0.9510565162951538);
\draw (-1.9270509831248428,-0.9510565162951534)-- (-2.0349766831264033,-1.283217666428072);
\draw (-1.969788670466215,0.8859650802544457) node[anchor=north west] {$u_2$};
\draw (-2.526428864551017,0.12610703753550965) node[anchor=north west] {$u_3$};
\draw (-2.0,-0.6) node[anchor=north west] {$u_4$};
\begin{scriptsize}
\draw [fill=black] (0.3090169943749475,0.9510565162951535) circle (1.5pt);
\draw [fill=black] (1.0,-0.0) circle (1.5pt);
\draw [fill=black] (-0.8090169943749473,0.5877852522924731) circle (1.5pt);
\draw [fill=black] (-0.8090169943749475,-0.5877852522924731) circle (1.5pt);
\draw [fill=black] (0.30901699437494745,-0.9510565162951536) circle (1.5pt);
\draw [fill=black] (0.41694269437650827,1.2832176664280721) circle (1.5pt);
\draw [fill=black] (1.3492549017242996,-0.0) circle (1.5pt);
\draw [fill=black] (0.4169426943765081,-1.2832176664280721) circle (1.5pt);
\draw [fill=black] (-2.618033988749895,0) circle (1.5pt);
\draw [fill=black] (-1.9270509831248428,0.9510565162951538) circle (1.5pt);
\draw [fill=black] (-1.9270509831248428,-0.9510565162951534) circle (1.5pt);
\draw [fill=black] (-2.618033988749895,0) circle (1.5pt);
\draw [fill=black] (-2.9672888904741948,0) circle (1.5pt);
\draw [fill=black] (-2.618033988749895,0) circle (1.5pt);
\draw [fill=black] (-2.034976683126403,1.2832176664280723) circle (1.5pt);
\draw [fill=black] (-2.0349766831264033,-1.283217666428072) circle (1.5pt);
\end{scriptsize}
\end{tikzpicture}
\end{center}
\caption{The construction of Lemma~\ref{bnoadj}. All the edges between the vertices $v_0$, $v_1$, $v_2$, $v_3$, $v_4$, $u_2$, $u_3$ and $u_4$ are represented.\label{bfigureplus}}
\end{figure}

\begin{lemm} \label{b33333}
There is no $5$-face in $G$ with only $3$-vertices.
\end{lemm}

\begin{proof}
Let $C = v_0v_1v_2v_3v_4$ be such a face, and let $u_i$ be the third neighbors of $v_i$ for $i \in \{0,1,2,3,4\}$. By Lemma~\ref{b33333linked444}, no more than two of the $u_i$ are $4$-vertices.

By the girth assumption, all the $u_i$ are distinct and two $u_i$ whose corresponding $v_i$ are adjacent are not adjacent. 

We prove now that there is no edge between the $u_i$. W.l.o.g. suppose $u_0u_2 \in E$. By Lemma~\ref{bsepcycle}, $u_0$ and $u_2$ are $4$-vertices. Let $H^* = G - C - \{u_0,u_2\}$. Graph $H^*$ has $n-7$ vertices and $m' \le m-14$ edges. Let $F'$ be any induced forest of $H^*$. Adding $v_0$, $v_1$, $v_2$ and $v_3$ to $F'$ leads to an induced forest of $G$. Observation~\ref{babg} applied to $(\alpha,\beta,\gamma) = (7,14,4)$ leads to a contradiction.

We now consider four cases:
\begin{itemize}
\item
Suppose two $u_i$ have degree $4$, and the corresponding $v_i$ are adjacent. W.l.o.g. $u_0$ and $u_1$ have degree $4$. 

Let us first assume that there is a vertex $w$ adjacent to $u_2$ and $u_3$. Vertex $w$ has degree $3$ by Lemmas~\ref{bsepcycle} and~\ref{b43333} (in particular $w \ne u_0$). Vertex $w$ is not adjacent to any of the $v_i$ or $u_i$ except for $u_2$ and $u_3$ by Lemma~\ref{bnoadj}.
Let $H^* = G - C - \{u_0,u_1,u_2,u_3,u_4,w\}$. Graph $H^*$ has $n-11$ vertices and $m' \le m-23$ edges. Adding $v_0$, $v_1$, $v_2$, $v_4$, $u_2$ and $u_3$ to any induced forest of $H^*$ leads to an induced forest of $G$. Observation~\ref{babg} applied to $(\alpha,\beta,\gamma) = (11,23,6)$ leads to a contradiction.

So there is no vertex $w$ adjacent to $u_2$ and $u_3$, and by symmetry there is no vertex $w$ adjacent to $u_3$ and $u_4$. By Lemma~\ref{b3facesadj} there is a vertex $w'$ adjacent to $u_4$ and $u_0$. By Lemmas~\ref{bsepcycle} and~\ref{b43333}, $w'$ has degree $4$. By Lemma~\ref{bdeg3deg4}, since there is no edge among the $u_i$ and by the girth assumption, there is a vertex $w$ adjacent to $u_3$ and $u_4$, a contradiction.

\item
Suppose two $u_i$ have degree $4$, and the corresponding $v_i$ are not adjacent. W.l.o.g. $u_0$ and $u_2$ have degree $4$. Then by Lemma~\ref{b3facesadj} there is a vertex $w'$ adjacent either to $u_0$ and $u_4$ or to $u_2$ and $u_3$. W.l.o.g. $w'$ is adjacent to $u_2$ and $u_3$. By Lemmas~\ref{bsepcycle} and~\ref{b43333}, $w'$ has degree $4$. By Lemma~\ref{bdeg3deg4}, since there is no edge among the $u_i$ and by the girth assumption, there is a vertex $w$ adjacent to $u_3$ and $u_4$. Vertex $w$ has degree $3$ by Lemmas~\ref{bsepcycle} and~\ref{b43333}. Vertex $w$ is not adjacent to any of the $v_i$ or $u_i$ except $u_3$ and $u_4$ by Lemma~\ref{bnoadj}. 
Let $H^* = G - C - \{u_0,u_1,u_2,u_3,u_4,w\}$. Graph $H^*$ has $n-11$ vertices and $m' \le m-23$ edges. Adding $v_0$, $v_1$, $v_2$, $v_4$, $u_3$ and $u_4$ to any induced forest of $H^*$ leads to an induced forest of $G$. Observation~\ref{babg} applied to $(\alpha,\beta,\gamma) = (11,23,6)$ leads to a contradiction.

\item
Suppose exactly one $u_i$ has degree $4$, $u_0$ w.l.o.g., and $u_0$ is adjacent to a vertex $w$ that is adjacent to either $u_1$ or $u_4$, say $u_1$. Vertex $w$ has degree $4$ by Lemmas~\ref{bsepcycle} and~\ref{b43333}. By Lemma~\ref{bdeg3deg4}, since there is no edge among the $u_i$ and by the girth assumption, there is a vertex $w'$ adjacent to $u_1$ and $u_2$. Moreover $w'$ has degree $3$ by Lemmas~\ref{bsepcycle} and~\ref{b43333}. Vertex $w'$ is not adjacent to any of the $v_i$ or $u_i$ except for $u_1$ and $u_2$ by Lemma~\ref{bnoadj}. By Lemma~\ref{b3facesadj}, there is a vertex $w''$ adjacent either to $u_2$ and $u_3$ or to $u_0$ and $u_4$.

Suppose $w''$ is adjacent to $u_2$ and $u_3$. By Lemmas~\ref{bsepcycle} and~\ref{b43333}, $w''$ has degree $3$, and $w''$ is not adjacent to any of the $v_i$ or $u_i$ except $u_2$ and $u_3$ by Lemma~\ref{bnoadj}. By the girth assumption, $w'w'' \notin E$ and $ww' \notin E$. By Lemmas~\ref{bsepcycle} and~\ref{b43333}, $ww'' \notin E$. By Lemma~\ref{bnoadj} applied to $v_2u_2w''u_3v_3$ and $v_1u_1w'u_2v_2$, $wu_3 \notin E$. Let $H^* = G - C - \{u_0,u_1,u_2,u_3,w,w',w''\}$. Graph $H^*$ has $n-12$ vertices and $m' \le m-23$ edges. Adding $v_0$, $v_2$, $v_4$, $u_1$, $u_2$, $u_3$, and $w'$ to any induced forest of $H^*$ leads to an induced forest of $G$. Observation~\ref{babg} applied to $(\alpha,\beta,\gamma) = (12,23,7)$ leads to a contradiction.

Thus $w''$ is adjacent to $u_0$ and $u_4$. By the same arguments as above, $w''$ being the symmetrical of $w$, $w''$ has degree $4$ and there is a $3$-vertex $w'''$ adjacent to $u_3$ and $u_4$, and not to any other of the $u_i$ and $v_i$. 

Suppose $w'w''' \in E$. Let $H^* = G - C - \{u_1,u_2,u_3,u_4,w',w'''\}$. Graph $H^*$ has $n-11$ vertices and $m' \le m-19$ edges. Adding $v_1$, $v_2$, $v_3$, $v_4$, $u_3$, $u_4$ and $w'$ to any induced forest of $H^*$ leads to an induced forest of $G$. Observation~\ref{babg} applied to $(\alpha,\beta,\gamma) = (11,19,7)$ leads to a contradiction.

Thus $w'w''' \notin E$. Recall that $w'$ and $w'''$ are not adjacent to any of the $v_i$ or $u_i$ except for $u_1$ and $u_2$, and $u_3$ and $u_4$ respectively. Let $H^* = G - C - \{u_0,u_1,u_2,u_3,u_4,w',w'''\}$. Graph $H^*$ has $n-12$ vertices and $m' \le m-23$ edges. Adding $v_0$, $v_1$, $v_2$, $v_3$, $u_3$, $u_4$ and $w'$ to any induced forest of $H^*$ leads to an induced forest of $G$. Observation~\ref{babg} applied to $(\alpha,\beta,\gamma) = (12,23,7)$ leads to a contradiction.

\item
Thus either all the $u_i$ have degree $3$, or $u_0$ has degree $4$ and there is no $w$ adjacent to $u_0$ and either to $u_1$ or to $u_4$. In both cases $u_1$, $u_2$, $u_3$ and $u_4$ have degree $3$, and, w.l.o.g., by Lemma~\ref{b3facesadj} there are vertices $w_1$, $w_2$ and $w_3$ adjacent to $u_1$ and $u_2$, to $u_2$ and $u_3$ and to $u_3$ and $u_4$ respectively. For all $j \in \{1,2,3\}$, by Lemmas~\ref{bsepcycle} and~\ref{b43333}, $w_{j}$ has degree $3$, and by Lemma~\ref{bnoadj}, $w_{j}$ is not adjacent to any of the $u_i$ and $v_i$ except for $u_{j}$ and $u_{j+1}$. We have $w_1w_2 \notin E$ and $w_2w_3 \notin E$ by the girth assumption, and $w_1w_3 \notin E$ by Lemma~\ref{bsepcyclefirst}. Let $H^* = G - C - \{u_0,u_1,u_2,u_3,u_4,w_1,w_2,w_3\}$. Graph $H^*$ has $n-13$ vertices and $m' \le m-23$ edges. Adding $v_0$, $v_1$, $v_2$, $v_3$, $u_1$, $u_2$, $u_3$ and $u_4$ to any induced forest of $H^*$ leads to an induced forest of $G$. Observation~\ref{babg} applied to $(\alpha,\beta,\gamma) = (13,23,8)$ completes the proof.
\end{itemize}
\end{proof}

Each $4$-vertex is in the boundary of at most four faces. Therefore the sum of the $c_{4}(f)$ over all the $5$-faces is $\sum_{f,l(f) = 5} c_{4}(f) \le 4n_4$.
From Lemmas~\ref{b43333} and~\ref{b33333} we can deduce that for each $5$-face $f$ we have $c_{4}(f) \ge 2$.  Thus $\sum_{f,l(f) = 5} c_{4}(f) \ge 2k_5$. Thus we have the following:

$$4n_4 \ge 2k_5$$
By Euler's formula, we have:

\begin{eqnarray*} 
-12 & = & 6m - 6n - 6k  \\
& = & 2\sum_{v \in V(G)}d(v) + \sum_{f \in F(G)}l(f) - 6n - 6k \\
& = & \sum_{d \ge 3}(2d - 6)n_d + \sum_{l \ge 5}(l-6)k_l  \\
& \ge & 2n_4 - k_5  \\
& \ge & 0 \\
\end{eqnarray*}
This is a contradiction, which ends the proof of Theorem~\ref{bgenmain}.

\bibliographystyle{plain}
\bibliography{biblio} {}

\end{document}